\documentclass[onecolumn,11pt]{article}
\usepackage{psfig}
\usepackage{amsmath}
\usepackage{geometry}                
\usepackage{graphicx}
\usepackage{amssymb}
\usepackage{epstopdf}
\usepackage{mathtools}
\usepackage{subcaption}
\usepackage{amsthm}
\usepackage {algorithm}
\usepackage {algpseudocode}
\DeclareMathOperator*{\argmax}{\arg\!\max}

\usepackage{epsfig}
\usepackage{url}

\begin{document}
\def\pbb#1{ \mathbb{P}\left\{  #1 \right\}}
%
%
%
%
%

%
%
%

\title{Fundamental Limits of Pooled-DNA Sequencing}
\author{
A.~Najafi
\thanks{Authors are with the Department of Computer Engineering and AICT Innovation Center, Sharif University of Technology, Tehran, Iran. (najafy@ce.sharif.edu and \{motahari,rabiee\}@sharif.edu)}
\and 
D.~Nashta-ali
\thanks{Authors are with the Department of Electrical Engineering and AICT Innovation Center, Sharif University of Technology, Tehran, Iran. (\{nashtaali,khani\}@ee.sharif.edu and khalaj@sharif.edu)}
\and 
S.~A.~Motahari
\footnotemark[1]
\and 
M.~Khani
\footnotemark[2]
\and 
B.~H.~Khalaj
\footnotemark[2]
\and 
H.~R.~Rabiee
\footnotemark[1]
}

\maketitle

\newtheorem{thm}{Theorem}[section]
\newtheorem{thm2}{Theorem}
\newtheorem{corl}[thm2]{Corollary}
\newtheorem{note}[thm2]{Note}
\newtheorem{definition}{Definition}
\newtheorem{lemma}{Lemma}

\begin{abstract}
In this paper, fundamental limits in sequencing of a set of closely related DNA molecules are addressed. This problem is called pooled-DNA sequencing which encompasses many interesting problems such as haplotype phasing, metageomics, and conventional pooled-DNA sequencing in the absence of tagging. From an information theoretic point of view, we have proposed fundamental limits on the number and length of DNA reads in order to achieve a reliable assembly of all the pooled DNA sequences. In particular, pooled-DNA sequencing from both noiseless and noisy reads are investigated in this paper. In the noiseless case, necessary and sufficient conditions on perfect assembly are derived. Moreover, asymptotically tight lower and upper bounds on the error probability of correct assembly are obtained under a biologically plausible probabilistic model. For the noisy case, we have proposed two novel DNA read denoising methods, as well as corresponding upper bounds on assembly error probabilities. It has been shown that, under mild circumstances, the performance of the reliable assembly converges to that of the noiseless regime when, for a given read length, the number of DNA reads is sufficiently large. Interestingly, the emergence of long DNA read technologies in recent years envisions the applicability of our results in real-world applications.
\end{abstract}

\section{Introduction}

Next Generation Sequencing (NGS) technologies paved the way for heterogeneous data extraction from different biological samples. Genomic datasets are of fundamental importance as DNA molecules are responsible to encode genetic instructions that are necessary for all the organisms to function. These datasets are used to recover the exact nucleotide composition of a set of DNA molecules, the so called DNA sequencing. Sequencing a set of DNA molecules can be categorized into two extreme cases. In the first case, the DNA molecules are evolutionary diverged enough such that it is safe to assume the problem is equivalent to sequencing a single DNA molecule with the total length of all the molecules. An evident example of such scenarios is the presence of numerous chromosomes within a cell which considerably differ from each other in terms of nucleotide content. In \cite{motahari2013information}, \cite{bresler2013optimal} and \cite{li1990towards} fundamental limits of sequencing of a single DNA molecule are obtained.

In the second case, DNA molecules are sampled from a population of closely related individuals where the molecules  mainly differ by Single Nucleotide Polymorphisms (SNPs). This paper aims  to obtain fundamental limits on the number and length of DNA  reads where exact reconstruction of all molecules is possible and reliable. Emergence of long DNA read technologies in the recent years has provided DNA reads with lengths up to almost $100$ kilo base-pair (kbp) \cite{laver2015assessing}, \cite{meyer2010illumina}. Using long DNA reads makes it possible for researchers to go beyond conventional limitations in a variety of computational methods in bioinformatics, such as haplotype phasing and DNA read alignment. Throughout this paper, sequencing of closely related DNA molecules is called pooled-DNA sequencing.

In pooled-DNA sequencing, the individual molecules may have been pooled by nature or by design. The former case can be approached by capturing each individual molecule for separate sequencing which is not a cost effective strategy \cite{cutler2010pool}. Even if the molecules are available separately, in the latter case, the molecules are pooled together to reduce the overall cost of sequencing. This can be done by reduction in overhead expenses through library preparation of a number of individuals simultaneously \cite{cao2016combinatorial}, \cite{meyer2010illumina}.

Pooled-DNA sequencing has diverse applications in modern biology and medicine. In cancer genomics, exact isolation of healthy tissues from cancerous cells can not be guaranteed during sample extraction. In sequencing of bacteria populations, different colonies cannot be easily separated, the so called metagenomics \cite{albertsen2013genome}. More importantly, in diploid individuals such as humans, somatic cells contain two homologous autosomal chromosomes that need to be sequenced, the so called haplotype phasing \cite{browning2011haplotype}.

Computational approaches to the problem of pooled-DNA sequencing has several advantages over experimental approaches \cite{rhee2016survey}, \cite{xu2002effectiveness}. On the computational side, a number of similar issues have been addressed in earlier works. Here, we consider the ones that are more relevant to this paper. Haplotype phasing is addressed in many papers, mostly through algorithmic approaches; see \cite{browning2011haplotype} for a review. In \cite{he2010optimal}, authors have shown that the optimal algorithm for haplotype assembly can be modeled as a dynamic programming problem. In \cite{delaneau2013improved}, authors have introduced SHAPEIT2, a statistical tool which adopts a Markov model to define the space of haplotypes which are consistent with given genotypes in a whole chromosome. A coverage bound based on the problem from the perspective of decoding convolutional codes is proposed in \cite{kamath2015optimal}. ParticleHap is also an algorithm proposed in \cite{ahn2015joint} to address the assembly problem with joint inference of genotypes generating the most likely haplotypes. Recently, another tool for the haplotype phasing problem using semi-definite programming is introduced in \cite{das2015sdhap}, which can be applied to polyploids as well. Many haplotype phasing methods proposed so far considered cohorts with nominally unrelated individuals, while in \cite{o2014general} a general framework for phasing of cohorts that contain some levels of relatedness is proposed. 

The aim of this paper is to show that it is feasible to reconstruct all the individual molecules if certain conditions hold. In particular, we will show that if the number and length of DNA reads are above some specific thresholds, which are functions of DNA sequence statistics, then reconstruction is possible.

This paper is organized as follows. 
In Section \ref{sec:mainmodel}, the mathematical model underlying the pooled-DNA sequencing problem  from DNA reads is introduced. We have presented a summary of our results and main contributions in Section \ref{sec:results}.
In Section \ref{sec:noiseless}, we have derived necessary and sufficient conditions for unique and correct assembly of all molecules in the noiseless regime. In this regard, analytic upper and lower bounds on the error probability of the reliable assembly  are obtained and asymptotic behaviors of the derived bounds are provided. Section \ref{sec:noisy} is devoted to analyze the problem of correct assembly from noisy DNA reads. Accordingly, we have formulated two upper bounds on the assembly error probability in the noisy regime via exploiting two novel denoising methods, i.e. maximum likelihood denoising and graph-based techniques motivated by community detection algorithms. Section \ref{sec:conclusion} concludes the paper.


\section{Basic Definitions and Notations}
\label{sec:mainmodel}

In this paper, we are interested in collective sequencing of a population of individuals from a given specie. It is desirable to identify the genome of each individual within the population using a single run of a sequencing machine, through the so-called pooled-DNA sequencing. In this section, we first provide a formal definition of the problem in addition to statistical models employed for analytic formulations.

\subsection{Population Structure}
We consider a population of $M$ distinct individuals where a reference genome of length $G$ is already sequenced and available. Genetic variations among individuals and the reference genome are assumed to be originated solely from Single Nucleotide Polymorphisms (SNPs). SNPs are a group of randomly spread nucleotides across the genome that differ in independent individuals with a high probability ($> 1\%$). It should be noted that the effect of Linkage Disequilibrium (LD) is ignored in this study and SNPs are assumed to be generated independently with respect to each other and also other individuals of the species. The frequency vector usually consists of only two non-zero elements, i.e. major and minor allele frequencies.

\subsection{Sequencing Method}
For the purpose of library preparation, $N$ distinct DNA fragments are extracted randomly and uniformly from each individual. In the current study, we simply assume DNA reads are not tagged and hence no information regarding the membership of fragments is available. All the fragments are pooled together and fed to a sequencing machine to produce single-end reads with length $L$. In Section \ref{sec:noiseless}, we assume sequencing machine does not produce any mismatch or indel errors while producing the reads. However, Section \ref{sec:noisy} assumes that DNA reads are altered with an error probability of $\epsilon$.


\subsection{Assembly Problem}
Recovering the genomes of all pooled individuals from the reference genome and $MN$ reads generated by the sequencing machine is the assembly problem. We would like to determine sufficient conditions for $L$ and $N$ such that for a given set of parameters $M$, $G$, $\epsilon$ and DNA sequence statistics the assembly of genomes belonging to each of $M$ individuals can be carried out uniquely and correctly. In this paper, we assume that all reads are mapped uniquely and correctly to the reference genome. Although this assumption is far from reality for short reads, it is quite valid for long and very long reads which are the main focus of this paper.


\subsection{ Probabilistic Modeling}\label{sec:model}

We exploit a biologically plausible probabilistic model to mathematically formulate our problem and assumptions. The latent independent SNP set of the $m$th individual is denoted as $\mathcal{S}_m=\left\{\mathcal{S}^{\left(1\right)}_m, \cdots, \mathcal{S}^{\left(S\right)}_m\right\}$ for all $m\in\left\{1, 2, \ldots, M\right\}$, where $S$ denotes the total number of SNPs. Also, $\mathcal{S}^{(s)}_m\in\left\{A,C,G,T\right\}$ for all $s$ and $m$. Moreover, we suppose that the loci of SNPs are known and randomly spread across the reference genome according to a {\it {Poisson}} process with rate $p$. Consequently, we have $S\simeq Gp$. The reason behind the {\it {Poisson}} assumption lies in the fact that SNP-causing mutations have happened uniformly and independently across the whole genome during the course of evolution. SNP values are assumed to be drawn from $S$ multinomial distributions with frequency vector $\mathbf{q}_s=\left(q_{s,A}, q_{s,C}, q_{s,G}, q_{s,T}\right)$, such that $q_{s,b}$ indicates the frequency of base $b\in\left\{A, C, G, T\right\}$ for the $s$th SNP. For the sake of simplicity, different SNPs in a single individual, and also all SNPs across different individuals are assumed to be statistically independent which indicates that  linkage disequilibrium and mutual ancestral properties are not taken into account.

The set of pooled DNA fragments are denoted by $\mathcal{R}=\left\{r_1,\ldots,r_{NM}\right\}$, where $r_m$ represents the $m$th sequenced DNA read. For all $r\in\mathcal{R}$, the read $r$ is assumed to be correctly aligned to the reference genome. In this regard, DNA read arrivals for each individual will be known and randomly spread across the genome according to a {\it {Poisson}} process with rate $\lambda=N/G$. Arrivals for different individuals are assumed to be independent from each other. Regarding the DNA reads, we will consider two regimes in this paper, i.e. noiseless and noisy regimes. In the former case, it has been assumed that DNA reads do not contain any sequencing error (mismatch or indel). In the noisy case, it has been assumed that sequencing system alters the nucleotide content of DNA reads. For simplicity, we assume that sequencing error does not produce any new alleles and simply transforms a minor allele into a major one and vice versa, both with probability of $\epsilon$, $0<\epsilon\leq 0.5$. In cases where an error event produces a new allele in a read, the  allele will be randomly mapped to either minor or major alleles.

The ultimate goal of this study is to investigate the conditions under which inference of  the hidden SNP sets for all individuals can be carried out correctly and without ambiguity. Also, we present a number of procedures for multiple genome assembly in such cases. Moreover, we have shown that there are conditions under which unique genome assembly is almost impossible regardless of the algorithms being employed.

\section{Summary of Results}
\label{sec:results}

\begin{figure}[t]
\centering
	\begin{subfigure}[b]{0.49\textwidth}
		\includegraphics[trim=0.7cm 0.4cm 0.8cm 0.8cm,clip,width=\textwidth]{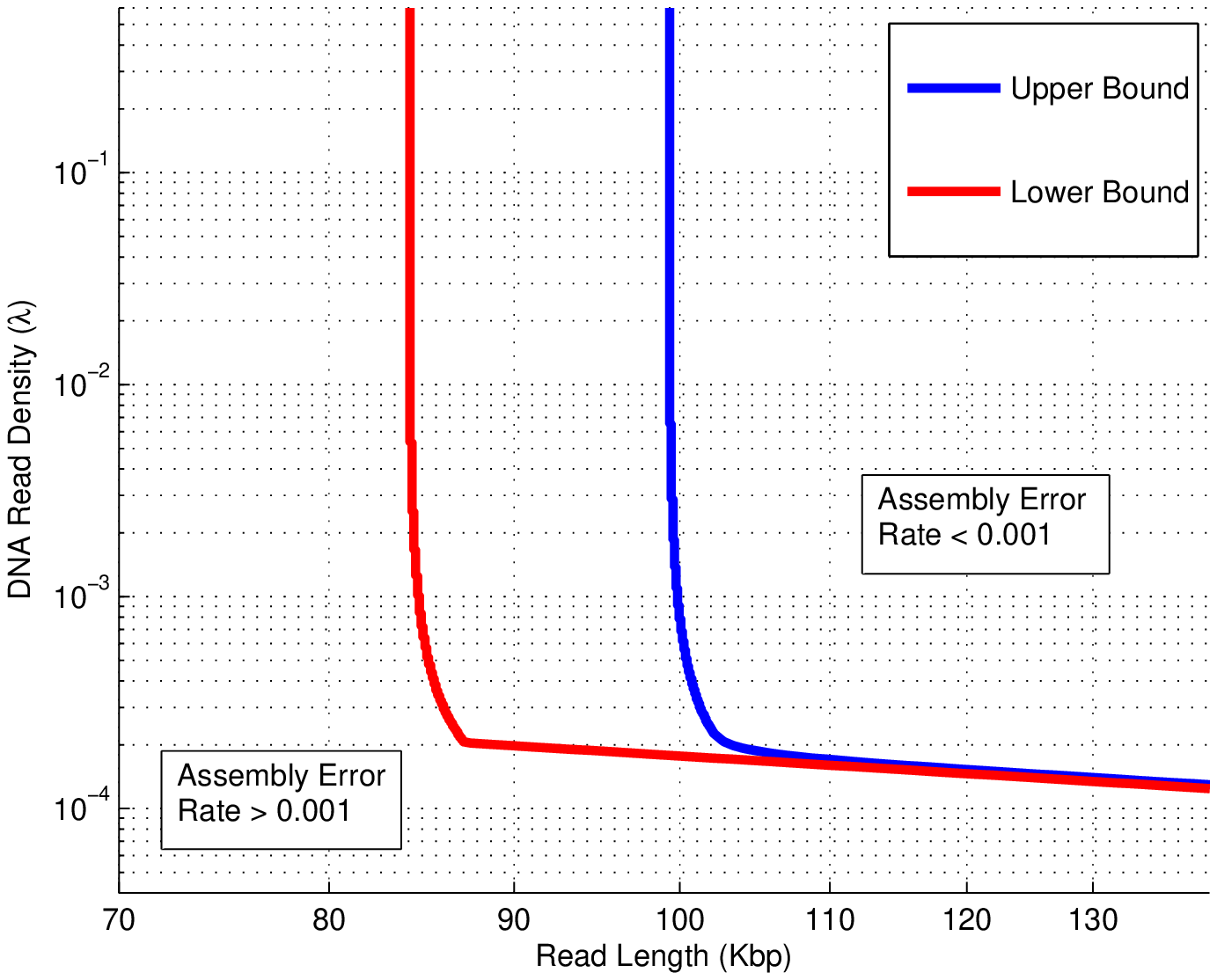}
		\caption{$M=2$ individuals}
		\label{fig:a:assemblyRegion}
	\end{subfigure}
	\begin{subfigure}[b]{0.49\textwidth}
		\includegraphics[trim=0.7cm 0.4cm 0.8cm 0.8cm,clip,width=\textwidth]{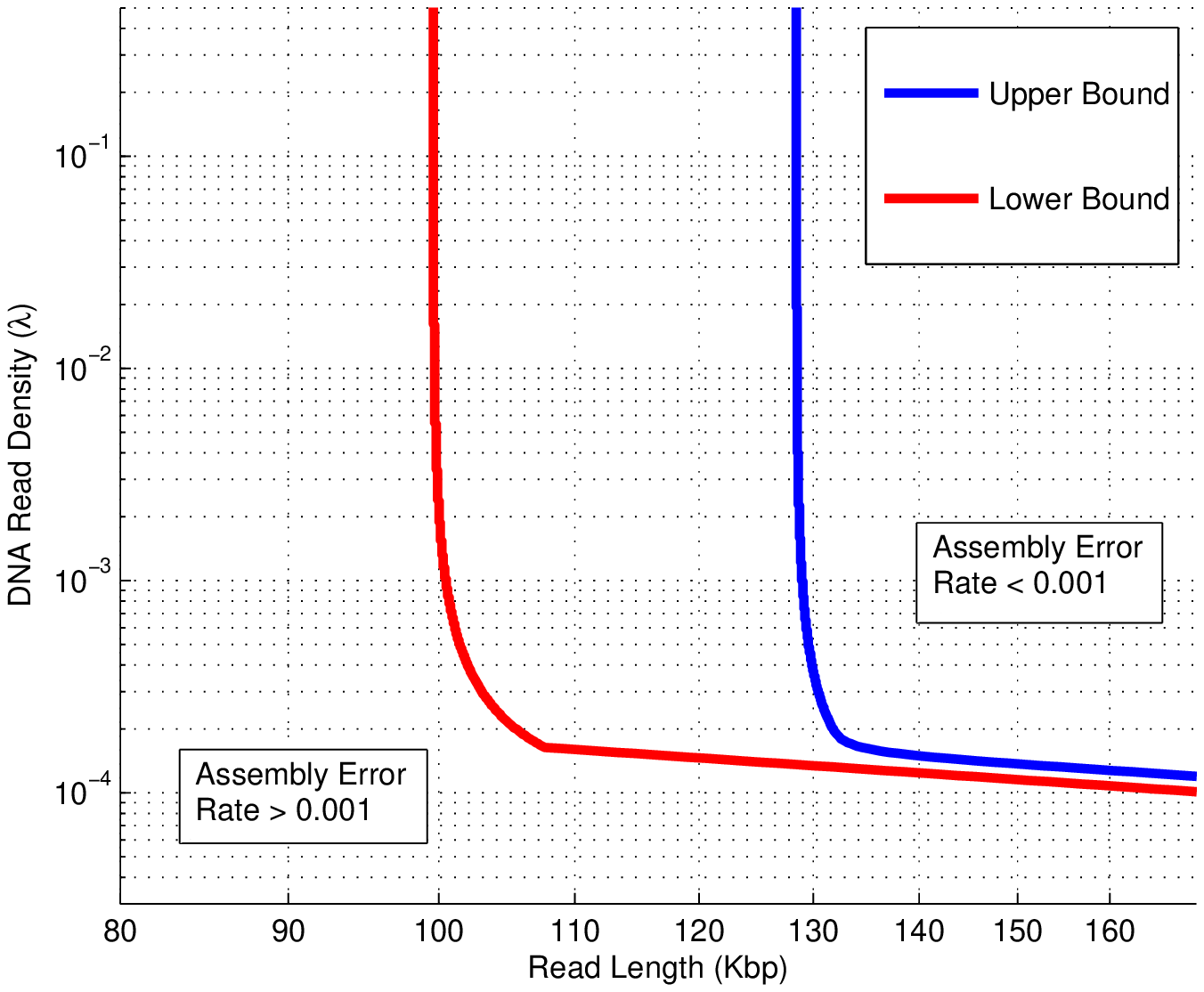}
		\caption{$M=20$ individuals}
		\label{fig:b:assemblyRegion}
	\end{subfigure}
	\caption{Assembly regions corresponding to (a) $M=2$, and (b) $M=20$ individuals, respectively. The areas beyond the lower bounds have an assembly error rate of equal or greater than $0.001$, while the regions beyond the upper bound guarantee an assembly error rate less then $0.001$.}
\label{fig:assemblyRegion}
\end{figure}
\begin{figure}[t]
\centering
	\includegraphics[trim=0cm 0cm 0cm 0cm,clip,width=0.8\textwidth]{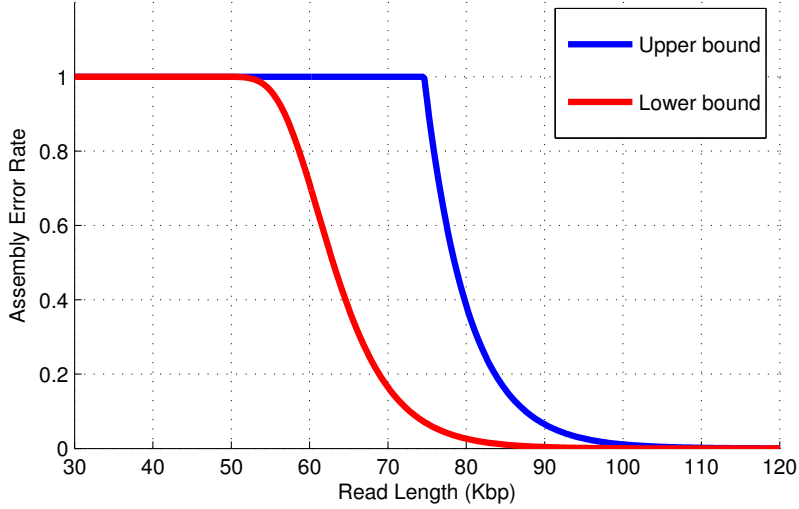}	
	\caption{Theoretical lower and upper bounds for assembly error rate in case of $M=2$ individuals. Sequencing depth is fixed regardless of read length $L$ and equals $\lambda L=43$. As can be seen, a sharp phase transition in assembly error rate is evident as the read length is being increased.}
\label{fig:phaseTransition}
\end{figure}

In the noiseless case, our main contribution includes necessary and sufficient conditions for unique and correct genome assembly of all individuals, in addition to respective lower and upper bounds on the assembly error probability. Let us denote $\mathcal{E}$ as the error event in genome assembly in the noiseless regime. Then, in Section \ref{sec:noiseless} it is shown that when DNA read density $\lambda$ is above some threshold and DNA read length $L$ is sufficiently large, the following inequalities hold:
\begin{equation}
\frac{G\left(M-1\right)}{L}
e^{-p\left(1-\eta\right)L}
\leq
\pbb{\mathcal{E}}
\leq
\frac{1}{2}GM^2p\left(1-\eta\right)e^{-p\left(1-\eta\right)L},
\end{equation}
where $\eta\triangleq\mathbb{E}_s\left\{\sum_{b\in\left\{A,C,G,T\right\}}q^2_{s,b}\right\}$ and is an attribute of genome statistics. Exact lower and upper bounds for the non-asymptotic scenario can be found in Section \ref{sec:noiseless} as well.

As genome length $G$ is increased while the number of individuals $M$ is fixed, the lower and upper bounds become tight and almost coincide with each other. However, for a practical configuration of parameters such as the ones given in Fig. \ref{fig:assemblyRegion}, the theoretical gap between lower and upper bounds becomes significant. In Fig. \ref{fig:assemblyRegion} we have demonstrated lower and upper bounds of the assembly region in $L$-$\lambda$ plain, where the bounds correspond to assembly error rate of $\pbb{\mathcal{E}}=0.001$. The number of individuals is assumed to be $M=2$ in Fig. \ref{fig:a:assemblyRegion}, and $M=20$ in Fig. \ref{fig:b:assemblyRegion}, respectively. In this simulation, we have assumed an SNP rate of $p=0.001$ and an effective minor allele frequency of $10\%$, i.e. $\eta=0.82$, where this numerical configuration corresponds to human DNA statistics. It is clear that when $\lambda\rightarrow\infty$, DNA read length $L$ must be above some critical length in order to ensure that the assembly error rate remains sufficiently small, i.e. $\pbb{\mathcal{E}}<0.001$. Evidently, for many other species of interest such as bacteria populations, the minimum required read length is significantly smaller due to their higher SNP rate $p$ and smaller genome length $G$.

Fig. \ref{fig:phaseTransition} shows the phase transition of the lower and upper bounds of the assembly error rate $\pbb{\mathcal{E}}$ when sequencing depth is fixed and read length $L$ is increased. The upper and lower bounds correspond to the case with $M=2$ individuals, and the sequencing depth is chosen as $\lambda L=43$ which is a common value in real-world applications. To summarize, the key observation is existence of a rather sharp phase transition as $L$ is increased in the case of both upper and lower bounds.

\begin{figure}[t]
\centering
	\includegraphics[trim=1.2cm 0.5cm 1.2cm 1.2cm,clip,width=0.8\textwidth]{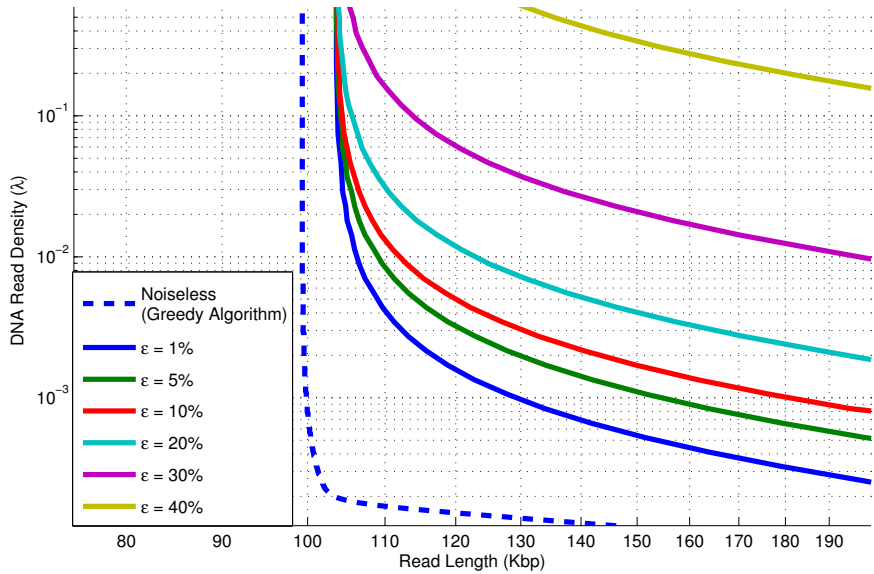}
	\caption{Upper bounds of assembly region in the noisy scenario, when maximum likelihood denoising algorithm is employed. Borders in $\lambda$-$L$ correspond to assembly error rate of $0.001$. Region upper bounds are depicted for sequencing error rates of  $1\%$, $5\%$, $10\%$, $20\%$, $30\%$ and $40\%$, respectively. Upper bound of assembly region corresponding to assembly error rate of $0.001$ for the noiseless case (greedy algorithm) is demonstrated for comparison. All sequencing error rates strictly smaller than $\epsilon=50\%$ can be handled via ML denoising as long as sufficient sequencing depths are provided.}
	\label{fig:MLAsemblyRegion}
\end{figure}
\begin{figure}[t]
\centering
	\includegraphics[trim=0.5cm 0cm 1cm 0.5cm,clip,width=0.7\textwidth]{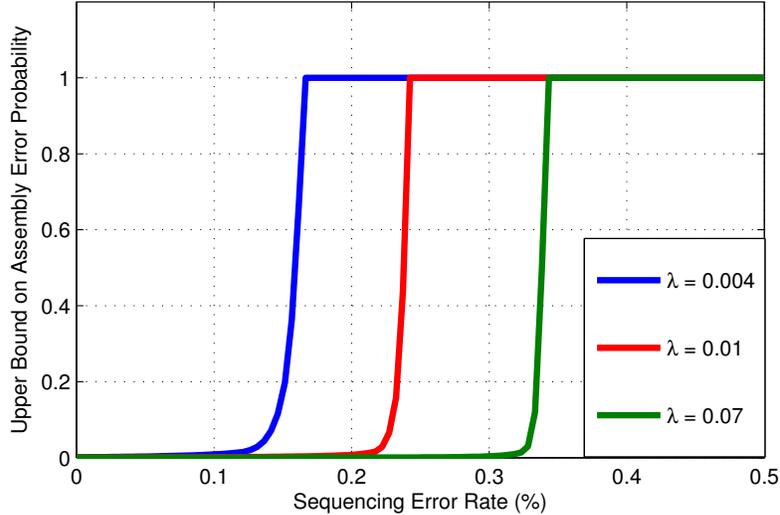}
	\caption{Upper bound on assembly error rate as a function of sequencing error rate $\epsilon$ for three DNA read density ($\lambda$) values and fixed DNA read length of $L=110$kbp. As can be seen drastic phase transitions can be observed as one increases the sequencing error. Large values of $\lambda$ correspond to further phase transitions.}
	\label{fig:MLNoisePerf}
\end{figure}

In the noisy regime, we have proposed a set of sufficient conditions for unique and reliable assembly of all the genomes. In addition, two denoising algorithms are proposed for block-wise inference of true genomic contents from noisy reads, denoted by Maximum-Likelihood (ML) and spectral denoising. In this regard, the following upper bound on the error probability of reliable genome assembly via ML denoising is derived for the case of $M=2$ individuals:
\begin{align}
\pbb{\mathcal{E}_N}\leq&
\inf_{D,d}~
\frac{G}{d}\left\{
e^{-p\left(1-\eta\right)\left(D-d\right)}
+
pMD
e^{-\frac{\lambda}{2}M\left(L-D\right)\left(1-2\sqrt{\epsilon\left(1-\epsilon\right)}\right)}
\right\}
\nonumber \\[2mm]
&~\text{subject to}~\quad0<d<D\leq L,
\label{eq:SORnoisy}
\end{align}
where $\mathcal{E}_N$ denotes the error event in assembly of all genomes in the noisy case. Exact non-asymptotic upper bound in the noisy case, in addition to an approximate analytical solution of the minimization problem in \eqref{eq:SORnoisy} and also the generalization of \eqref{eq:SORnoisy} for any $M\ge2$ can be found in Section \ref{sec:noisy}. Moreover, a similar upper bound and an approximate mathematical analysis is presented for the case of spectral denoising. Our proposed spectral method has a similar performance to that of ML method when $\epsilon$ is small, however, it is much more computationally efficient.

\begin{figure}[t]
\centering
	\includegraphics[trim=0.7cm 0.2cm 1cm 0.7cm,clip,width=0.7\textwidth]{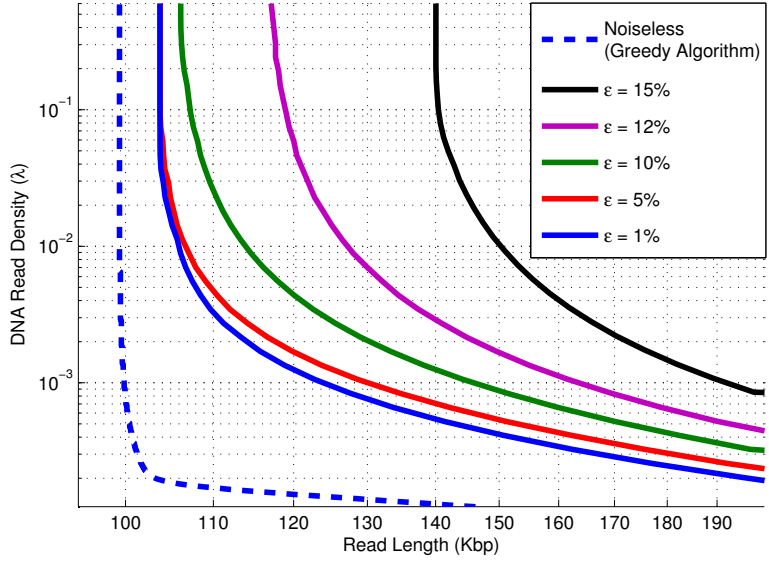}
	\caption{Upper bounds of assembly region in the noisy scenario, when spectral denoising algorithm is employed. Borders in $\lambda$-$L$ plain correspond to assembly error rate of $\pbb{\mathcal{E}_N}\leq 0.001$. Upper bounds of assembly region are depicted for sequencing error rates of  $1\%$, $5\%$, $10\%$, $12\%$ and $15\%$, respectively. Also, upper bounds corresponding to the same assembly error rate for the noiseless case (greedy algorithm) is demonstrated for comparison. As can be verified, sequencing error rates more than $20\%$ cannot be handled by spectral denoising when DNA read lengths are below $200$kbp.}
	\label{fig:SDAssemblyRegion}
\end{figure}

\begin{figure}[t]
\centering
	\includegraphics[trim=0.5cm 0cm 1cm 0.5cm,clip,width=0.7\textwidth]{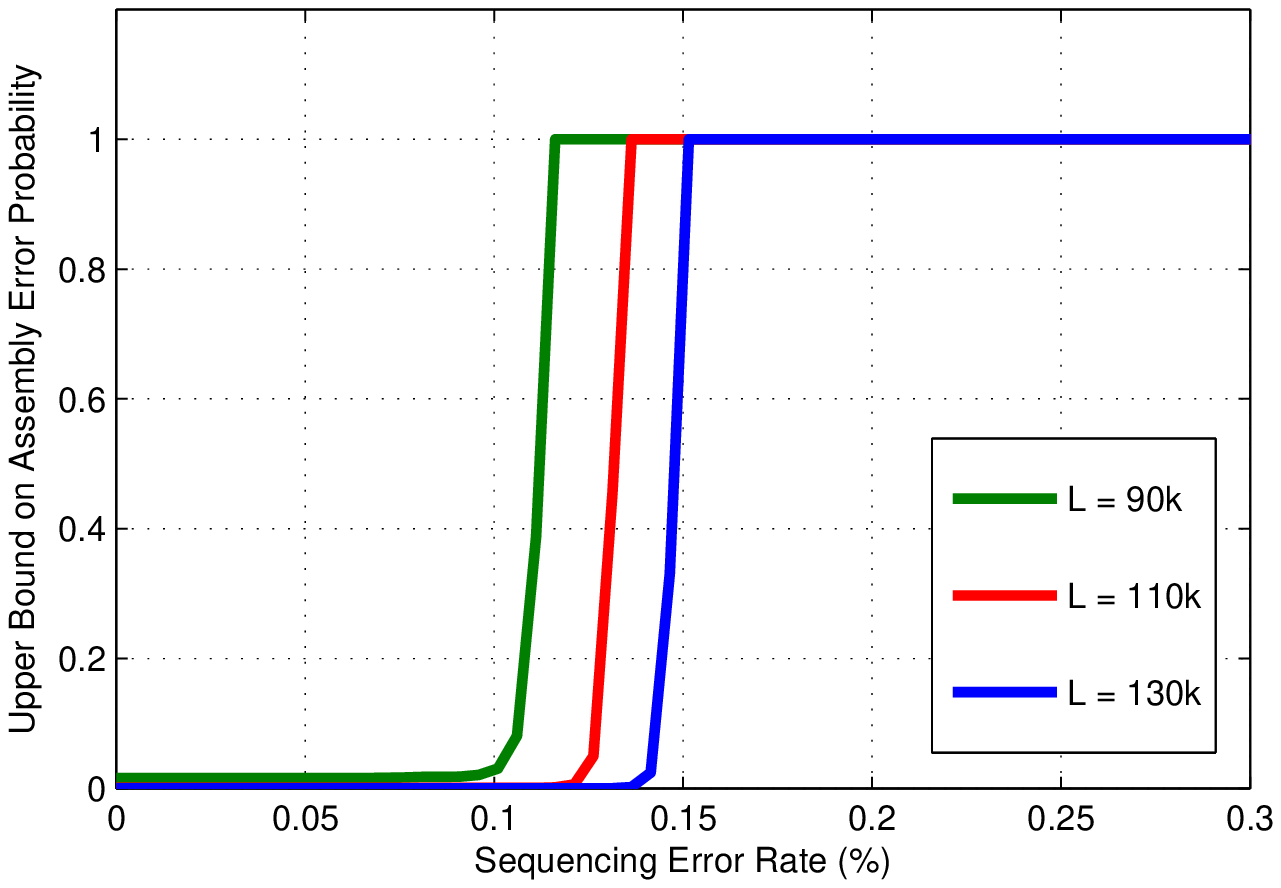}
	\caption{Upper bound on assembly error rate as a function of sequencing error probability $\epsilon$ for three DNA read lengths $L$ and fixed DNA read density of $\lambda=0.01$. Spectral denoising algorithm is employed as the denoising method. As can be verified, drastic phase transitions can be observed as one increases the sequencing error. Large values of DNA read length correspond to occurrences of phase transitions at larger values of $\epsilon$.}
	\label{fig:SDNoisePerf}
\end{figure}

In Fig. \ref{fig:MLAsemblyRegion}, upper bounds on the assembly region in $L$-$\lambda$ plain are shown for the assembly error rate of $\pbb{\mathcal{E}_N}\le 0.001$, when ML denoising is employed. Results are shown for six different sequencing error rates. Other configurations such as $\eta$, $p$ and $G$ are identical to those used for Fig. \ref{fig:assemblyRegion}  which resemble the human genetic settings. As can be seen, regardless of the sequencing error rate, by increasing $\lambda$ the upper bounds converge to a specified DNA read length which is close to that of the noiseless case (sufficient condition for the greedy algorithm).

Fig. \ref{fig:MLNoisePerf} demonstrates the upper bounds on assembly error probability as a function of sequencing error rate for three different DNA read densities. Again, ML technique is  employed for denoising of reads. DNA read length is fixed to $L=110$kbp for all values of $\lambda$. It can be seen that upper bounds undergo a drastic phase transition as sequencing error is increased. However, larger values of $\lambda$ correspond to shifting the phase transition into larger sequencing error rates. It should be noted that for both Fig. \ref{fig:MLAsemblyRegion} and Fig. \ref{fig:MLNoisePerf} the number of individuals is assumed to be $M=2$.

Fig. \ref{fig:SDAssemblyRegion} shows the upper bounds of assembly region in $L$-$\lambda$ plain for different sequencing error rates, when spectral denoising is employed. Upper bounds indicate maximum assembly error rate of $0.001$. Other parameters such as $p$, $G$ and $\eta$ are assumed to be the same as those presumed in Fig. \ref{fig:assemblyRegion} for the human genetic settings. Upper bounds corresponding to $0\leq\epsilon\leq 5\%$ are very close to each other and also to the upper bound associated with the noiseless regime (greedy algorithm). As the sequencing error is increased, assembly region is shifted toward larger values of $\lambda$ and $L$. Moreover, when $\epsilon>20\%$ the upper bound for the maximum assembly error rate of $0.001$ cannot be satisfied by $L<200$kbp.

Fig. \ref{fig:SDNoisePerf} demonstrates the upper bounds on assembly error probability via spectral denoising as a function of sequencing error rate for three different DNA read lengths. DNA read density is assumed to be fixed and equal to $\lambda=0.01$ in all graphs. As can be seen, similar to ML denoising, drastic phase transitions can be observed as one increases the sequencing error which can be shifted toward larger values of $\epsilon$ provided that $L$ is chosen sufficiently large.

\section{Assembly Analysis: Noiseless Reads}
\label{sec:noiseless}

In this section, pooled-DNA sequencing of $M$ individuals for the simple case of tagless and noiseless DNA reads is studied. For any assembling strategy, one can define the error event as the event of failure in reconstructing the genomes of all individuals {\it {uniquely}} and {\it {correctly}}. This error event is denoted by $\mathcal{E}$. We aim to analytically compute lower and upper bounds for $\mathbb{P}\left\{\mathcal{E}\right\}$ based on the statistical models given in Section \ref{sec:model}. To this end, we first obtain necessary and sufficient conditions of unique and correct assembly of all $M$ individuals. The conditions, then, are used to derive tight bounds on $\mathbb{P}\left\{\mathcal{E}\right\}$. 

\begin{figure}[t]
\centering
	\includegraphics[trim=0cm 5cm 0cm 3cm,clip,width=\textwidth]{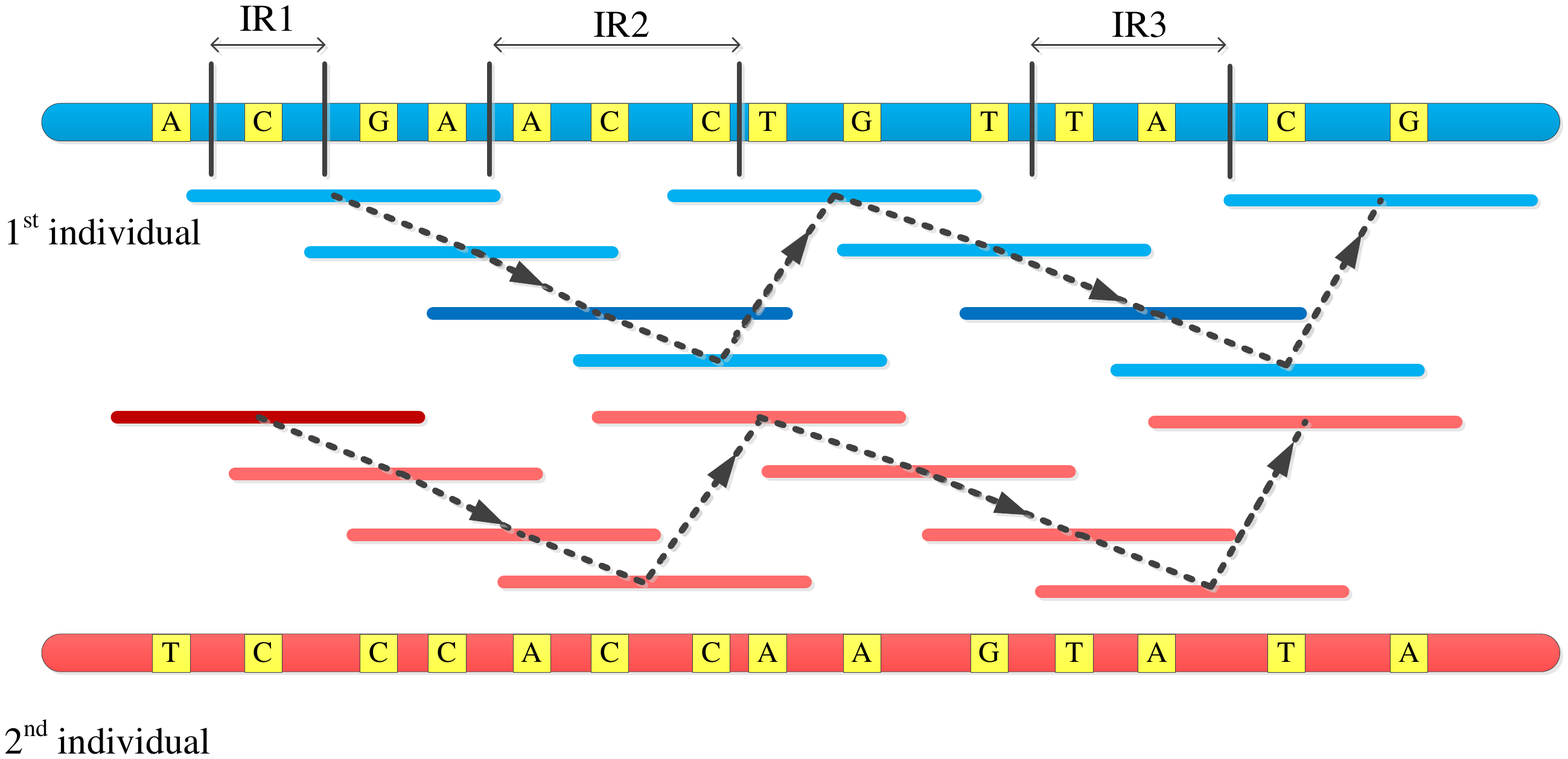}
	\caption{A sample pooled-DNA sequencing scenario with two individuals. Differences between genomes are solely caused by SNPs. As can be easily verified, all SNPs of both individuals are covered. Moreover, identical regions IR$1$, $2$ and $3$ in addition to all inter-SNP segments are bridged. Consequently, one can uniquely and correctly assemble the two genomes.}
\label{fig:bridge}
\end{figure}

In Theorem \ref{thm:MainTheorem}, we show that two conditions are enough for necessity and sufficiency of correct assembly, namely {\it{coverage}} and {\it{bridging}} conditions. The coverage condition is met whenever every single SNP in each individual is covered by at least one DNA read from that individual. The bridging condition is met whenever every single {\it {identical region}} between each pair of individuals is bridged by at least one DNA read from either of the individuals. An {\it{identical region}} between two particular individuals refers to any segment of genome in which they possess completely identical genomic content. If a DNA read starts before an identical region and ends after it, that region is said to be bridged. Fig. \ref{fig:bridge} demonstrates the concept of bridging and coverage in a pooled-DNA sequencing of two individuals. Blue and red DNA reads are associated with the first and second individuals, respectively. Bold reads have bridged the three long identical regions in addition to the identical regions corresponding to inter-SNP segments. The following theorem mathematically formulates the necessary and sufficient conditions for unique and correct genome assembly:

\begin{thm2}[Necessary and Sufficient Conditions for Unique Genome Assembly]
\label{thm:MainTheorem}
In a pooled-DNA sequencing scenario with tagless and noiseless DNA reads, the following conditions are necessary and sufficient for unique and correct genome assembly:
\begin{itemize}
\item SNP Coverage Condition (SC): every SNP must be covered by at least one DNA read from each individual,
\vspace*{-1mm}
\item Bridging Condition (B):  every identical region between any two individuals must be bridged by at least one DNA read from either of the individuals.
\end{itemize}
\end{thm2}

\begin{proof}[Proof of Necessity]
We will  prove that if any of the conditions mentioned in Theorem \ref{thm:MainTheorem} does not hold, then unique and correct assembly becomes impossible. First, assume at least one SNP is not covered by any DNA read for at least one individual. Then, due to lack of information regarding the value of that allele for at least one individual and the final assembly is not unique. Second, assume an identical region between two particular individuals is not bridged. Even if it is possible to correctly assemble both sides of the unbridged identical region, absence of any bridging read inhibits the flow of  membership information from one side of the region to the other. As a result, at least two completely distinct and yet legitimate genomes can be assembled which contradicts the required uniqueness of assembly.
\end{proof}

\begin{proof}[Proof of Sufficiency: The Greedy Algorithm]
It is desirable to design an algorithm that assembles the genomes efficiently and correctly. We will prove that a simple greedy algorithm can assemble the genomes, if the conditions provided in Theorem \ref{thm:MainTheorem} are met. To this end, we first sort the list of reads according to their starting positions. This ordering is possible because it is assumed that reads are mapped correctly and uniquely to the reference genome. Let us denote the sorted read set by $\mathcal{R} = \left\{r_1, r_2, \cdots, r_{N_{\text {tot}}} \right\}$. The proposed algorithm is detailed in Algorithm \ref{alg:greedy}. 

\medskip
\begin{algorithm}
\caption{The greedy algorithm for assembly.}\label{alg:greedy}
\vspace{.1cm}

\textbf{Inputs}
\begin{algorithmic}
\State{$N_{\text {tot}}$ aligned fragments of length $L$ from all individual genomes of length $G$.}
\State{Set of $\mathcal{R}=\{r_1,\cdots,r_{N_{\text {tot}}}\}$ for sorted set of all reads by their start points.}
\end{algorithmic}

\textbf{Output}
\begin{algorithmic}
\State{Assembled genomes of $M$ individuals: $\{C_{1},\cdots, C_{M}\}$.}
\end{algorithmic}

\vspace{.1cm}
\hrulefill

\textbf{Initialization}

\begin{algorithmic}[1]
\For {$ i = 1$ to $M$ }
\State {$C_i = \emptyset$}
\EndFor
\end{algorithmic}
\vspace{.1cm}
\hrulefill

\textbf{Greedy Merging}

\begin{algorithmic}[1]
\For {$ i = 1$ to $N_{\text {tot}}$ }
\State {Find the set $\mathcal{I} \subset \{C_{1},\cdots, C_{M}\}$ such that $r_i$ can be consistently merged to each member of $\mathcal{I}$.}
\State{Merge $r_i$ with $C_j\in \mathcal{I} $ if they share the maximum number of SNPs. If there are more than one candidate, choose one at random.}
\EndFor
\end{algorithmic}
\vspace{.1cm}
\end{algorithm}
\medskip

We next show that the greedy algorithm assembles all the genomes correctly given the conditions in Theorem \ref{thm:MainTheorem} are met. The proof is by induction. Assume that the algorithm is correct before merging read $r_i$, i.e., the partially constructed genomes $\{C_{1},\cdots, C_{M}\}$, referred to  as contigs, correspond to the true genomes and all reads prior to $r_i$ are merged correctly to their corresponding contigs. Without loss of generality, we assume $r_i$ is a read from the first genome. The algorithm fails if there exists $C_j \in \mathcal{I}$ with $j\neq 1$ with overlap size greater than or equal to that of $C_1$. We will show that this event does not happen if the conditions in Theorem \ref{thm:MainTheorem} are met. Let $b_j$ denote the overlap size between $C_j$ and $r_i$ for all $C_j\in\mathcal{I}$. We need to show that $b_j< b_1$ for all $j\neq 1$. Assume $b_j \geq b_1$. This case corresponds to the existence of an identical region between first and $j$'th individuals. From the third condition of the theorem, this identical region is bridged by at least one read from either individuals. The bridging read starts earlier than $r_i$ and based on our assumption is correctly merged to either $C_1$ or $C_j$. If the read comes from the first individual, then the overlap size between $C_j$ and $r_i$ is always less that that of $C_1$ and $r_i$ which contradicts the assumption. On the other hand, if the read comes from the $j$'th individual, $C_j$ cannot be in $\mathcal{I}$ as it is not consistent with $r_i$. This completes the proof.
\end{proof}

An implication of Theorem \ref{thm:MainTheorem} is that the error event $\mathcal{E}$ is the union of two events, denoted by $\mathcal{E}_{SC}$ and $\mathcal{E}_{B}$, which are due to SNP Coverage and Bridging conditions, respectively. Consequently, one can obtain lower and upper bounds on $\pbb{\mathcal{E}}$ as
\begin{equation}
\max\left\{
\mathbb{P}\left\{\mathcal{E}_{SC}\right\},
\mathbb{P}\left\{\mathcal{E}_B\right\}
\right\}
\leq
\mathbb{P}\left\{\mathcal{E}\right\}
\leq
\mathbb{P}\left\{\mathcal{E}_{SC}\right\}+
\mathbb{P}\left\{\mathcal{E}_B\right\}.
\label{eq:totErrorBound}
\end{equation}

In the following subsections, we first analyze the probability of $\mathcal{E}_{SC}$ and $\mathcal{E}_{B}$ in terms of $\lambda$ and $L$. In particular, we obtain the exponents of decay of each of these probabilities when one increases $\lambda$ and $L$. These exponents are then used to  obtain asymptotic bounds on $\mathbb{P}\left\{\mathcal{E}\right\}$.

\subsection{SNP Coverage Condition}

The first condition in Theorem \ref{thm:MainTheorem} states that each SNP should be contained in at least one read from each of $M$ individuals. Equivalently, for each individual, the starting point of at least one read should fall within the distance $L$ of each SNP. As DNA fragments are randomly spread along genome with the density of $\lambda=N/G$, the distance of two consecutive fragments has an exponential distribution of the form $\lambda e^{-\lambda\ell}$. Let $\mathcal{E}^{\left(j\right)}_{SC}$ denote the event that there exists at least on SNP in the $j$th individual which is not covered by any read. For any $j$, the probability of occurring $\mathcal{E}^{\left(j\right)}_{SC}$ conditioned on the number of reads can be exactly written as:
\begin{equation}
\mathbb{P}\left\{\mathcal{E}_{SC}^{\left(j\right)}\vert N\right\} =
1-\left(1- \frac{p}{p+\lambda} e^{-\lambda L}\right)^{N},
\end{equation}
where $N$ has a {\it {Poisson}} distribution with average $\lambda G$. Therefore,
\begin{align}
\mathbb{P}\left\{\mathcal{E}_{SC}^{\left(j\right)}\right\}=
\mathbb{E}\left\{\mathbb{P}\left\{\mathcal{E}^{\left(j\right)}_{SC}\vert N\right\}\right\}
&=
e^{-\lambda G}\sum_{N=0}^{\infty}\frac{\left(\lambda G\right)^N}{N!}
\left(
1-\left(1- \frac{p}{p+\lambda} e^{-\lambda L}\right)^{N}
\right)
\nonumber \\
&=
1-\exp\left(
\frac{-Ge^{-\lambda L}}{\frac{1}{p}+\frac{1}{\lambda}}
\right).
\end{align}
This implies that if one wishes the SNP Coverage condition to hold for a single individual with probability $1-\epsilon$ (for $\epsilon\ll 1$ and $\lambda G\gg1$), then read length $L$ should be chosen close to:
$$ L\simeq\frac{1}{\lambda} \log \left(\frac{G}{\left(\frac{1}{p}+\frac{1}{\lambda}\right)\epsilon}  \right).$$

Instead of obtaining exact formulation for the case of $M$ individuals, we derive tight upper and lower bounds on $\pbb{\mathcal{E}_{SC}}$. The upper bound can be simply attained via the union bound as follows:
\begin{equation}
\pbb{\mathcal{E}_{SC}}\leq
M\left(
1-\exp\left(
\frac{-Ge^{-\lambda L}}{\frac{1}{p}+\frac{1}{\lambda}}
\right)
\right).
\end{equation}

For the lower bound, we have derived two asymptotically tight formulations, where one formulation is appropriate for large $M$ and the other is tight for large $L$. When the number of individuals, $M$, is not very large, one can bound $\pbb{\mathcal{E}_{SC}}$ from below simply by considering the fact that $\pbb{\mathcal{E}_{SC}}\ge\pbb{\mathcal{E}_{SC}^{\left(j\right)}}$. Therefore, we have:
\begin{equation}
\pbb{\mathcal{E}_{SC}}\ge
1-\exp\left(
\frac{-Ge^{-\lambda L}}{\frac{1}{p}+\frac{1}{\lambda}}
\right).
\label{eq:exactLower1}
\end{equation}
This formulation is not tight when the number of individuals is very large. Therefore, we offer another lower bound  which linearly grows with $M$. Let us divide the whole genome into $G/\left(L+x\right)$ non-overlapping segments of length $L+x$. If one individual does not have a read starting within the segment, then there exists an interval with length $x$ which is not covered by all individuals. Moreover, if there exists one SNP which arrives within this interval, then the coverage condition does not hold. In this way, $\pbb{\mathcal{E}_{SC}}$ can be lower bounded as:
\begin{align}
\pbb{\mathcal{E}_{SC}}&\ge
\sup_{x}~
\sum_{k=0}^{\frac{G}{L+x}}\binom{\frac{G}{L+x}}{k}P_x^k\left(1-P_x\right)^{n-k}
\left(1-e^{-pkx}\right)=
\nonumber \\
&=1-
\inf_{x}
\left(1-P_x\left(1-e^{-px}\right)\right)^{\frac{G}{L+x}},
\label{eq:SCLowerBound}
\end{align}
where $k$ denotes the total number of intervals not covered by all individuals. In addition, $P_x$ represents the probability of arriving DNA reads belonging to less than $M$ individuals in a single segment of length $L+x$, and can be formulated as:
\begin{equation}
P_x\triangleq 1-\left(1-e^{-\lambda\left(L+x\right)}\right)^M.
\end{equation}
It can be shown that the optimal value of $x$ which maximizes the lower bound does not have a closed form analytic formulation. However, asymptotic analysis of \eqref{eq:SCLowerBound} for the case of $L\gg 1/\lambda$ results in a simpler mathematical formulation whose maximal point is analytically tractable and provides a good approximation for the optimal value of $x$, denoted by $x_{\text {opt}}$. For large values of $L$, the lower bound can be simplified as:
\begin{equation}
\pbb{\mathcal{E}_{SC}}\ge
GMe^{-\lambda L}
\sup_x
\frac{e^{-\lambda x}\left(1-e^{-px}\right)}{L+x},
\label{eq:lowerBoundAsymptot}
\end{equation}
and it is easy to show that the maximizer of the above inequality can be closely approximated by
\begin{equation}
x_{\text {opt}}\simeq \frac{1}{p}\log\left(1+\frac{p}{\lambda}\right).
\nonumber
\end{equation}
Substitution of this approximate maximizer in \eqref{eq:lowerBoundAsymptot} results in the following lower and upper bounds for $\pbb{\mathcal{E}_{SC}}$ in the asymptotic regime, where $\lambda L\gg 1$:
\begin{equation}
\frac{GM}{\frac{1}{p}+\frac{1}{\lambda}}e^{-\lambda L}~
\max\left\{
\frac{1}{M},
\frac{
\left(1+\frac{p}{\lambda}\right)^{\frac{-\lambda}{p}}
}
{\lambda L + 
\frac{\lambda}{p}\log\left(
1+\frac{p}{\lambda}
\right)
}
\right\}
\leq
\pbb{\mathcal{E}_{SC}}
\leq
\frac{GM}{\frac{1}{p}+\frac{1}{\lambda}}e^{-\lambda L}
\end{equation}

The results for the non-asymptotic case can be easily computed by substituting $x_{\text {opt}}$ into \eqref{eq:SCLowerBound}, since any value of $x$ implies a lower bound. From the asymptotic analysis it is evident that both lower and upper bounds have the same exponent of decay with respect to $L$. Moreover, the duality in the formulation of lower bound yields that for small values of $M$, the error probability of SNP coverage in a single individual is an acceptable approximation for $\pbb{\mathcal{E}_{SC}}$. On the other hand, when $M$ is large, the alternate lower bound is more favorable. More precisely, when $\lambda\gg p$, the transition between the two lower bounds occur when $M\simeq \lambda Le$.
\subsection{Bridging Condition}

In this subsection, we first obtain lower and upper bounds on $\mathbb{P}\left\{\mathcal{E}_B\right\}$ for the case of two individuals ($M=2$). Interestingly, it is shown that exponents of decay of both lower and upper bounds are the same with respect to read length $L$ and read density $\lambda$. We also provide an exact analysis of the error probability in Appendix \ref{app:exactBridging} which can be used for numerical analyses. Generalization to more than two individuals is carried out at the end of this subsection. 

\subsubsection{Case of Two Individuals}
Let us denote a SNP which has different alleles in the two individuals as a discriminating SNP. Identical regions between the two genomes correspond to segments between any two consecutive discriminating SNPs. The bridging condition for the two individual case implies that all of these identical regions must be bridged by at least one read from one of the two individuals. In the following, we first show that under the assumptions of Section \ref{sec:model} the arrival of discriminating SNPs corresponds to a {\it {Poisson}} process. We subsequently derive lower and upper bounds on $\pbb{\mathcal{E}_B\vert M=2}$ based on this property.

Per definition, the difference between genomes of two particular individuals is solely due to discriminating SNPs. The $s$th SNP, where $s\in \{1,2,\ldots,S\}$, has identical alleles in the two individuals with a probability of:
\begin{equation}
\eta_s=\sum_{b\in\left\{A,G,C,T\right\}}q^2_{s,b}.
\end{equation}
It has been further assumed that each allele frequency vector $\boldsymbol{q}_s=\left(q_{s,A},q_{s,C},q_{s,G},q_{s,T}\right)$ is sampled independently from a fixed and known multi-dimensional distribution $\pbb{\boldsymbol{q}}$. For a given segment of length $\ell$, if $n$ is the total number of SNPs falling within this segment, then it has a {\it {Poisson}} distribution of the form:
\begin{equation}
\pbb{n}=\frac{\left(p\ell\right)^n}{n!}e^{-p\ell}~,~n=0,1,\ldots.
\nonumber
\end{equation}
Similarly, if $n_d$, $0\leq n_d\leq n$, is the total number of discriminating SNPs in this segment, then its distribution can be given by:
\begin{equation}
\pbb{n_d} = \mathbb{E}_s\left\{
\sum_{n=n_d}^{\infty}
\frac{\left(p\ell\right)^n}{n!}e^{-p\ell}
\left(\sum_{\mathcal{R}_{n_d}}
\prod_{s\in\mathcal{R}_{n_d}}\eta_s
\prod_{s\notin\mathcal{R}_{n_d}}\left(1-\eta_s\right)
\right)
\right\},
\end{equation}
where $\mathcal{R}_{n_d}$ represents the set of all possible subsets of $\left\{1,2,\ldots,n\right\}$ with cardinality $n_d$. Since the sequence of $\eta_s,s=1,2,\ldots,S$ is i.i.d. and $\left\vert\mathcal{R}_{n_d}\right\vert=\binom{n}{n_d}$, we will have:
\begin{equation}
\pbb{n_d} =\sum_{n=n_d}^{\infty}
\frac{\left(p\ell\right)^n}{n!}e^{-p\ell}
\binom{n}{n_d}
\eta^{n_d}
\left(1-\eta\right)^{n-n_d}=
\frac{\left(p\left(1-\eta\right)\ell\right)^{n_d}}{n_d!}e^{-p\left(1-\eta\right)\ell},
\end{equation}
where $\eta=\mathbb{E}_s\left\{\eta_s\right\}$. Thus, the set of discriminating SNPs can also be modeled by a {\it {Poisson}} process with rate $p\left(1-\eta\right)$.

We first focus on the error probability of bridging a single identical region. This event is denoted $\mathcal{E}^{\left(1\right)}_B$.  One can show that:
\begin{align}
\mathcal{P}_2
\triangleq
\mathbb{P}\left\{\mathcal{E}^{\left(1\right)}_B \vert M=2\right\}
&=
e^{-p\left(1-\eta\right)L}+
\int_{0}^{L}p\left(1-\eta\right)e^{-p\left(1-\eta\right)\ell}
e^{-2\lambda\left(L-\ell\right)}{\text d}\ell
=
\nonumber \\
&=\left\{\begin{array}{cc}
\frac{2\lambda e^{-p\left(1-\eta\right)L} - p\left(1-\eta\right)e^{-2\lambda L}}{2\lambda-p\left(1-\eta\right)}
&,~  p\left(\frac{1-\eta}{2}\right)\neq \lambda
\\[4mm]
\left(1+Lp\left(1-\eta\right)\right)e^{-p\left(1-\eta\right)L}
&,~   p\left(\frac{1-\eta}{2}\right)= \lambda.
\end{array}
\right.
\end{align}

For sufficiently large values of $L$, the parameter $\mathcal{P}_2$ appears in our lower and upper bounds explicitly as:
\begin{equation}
\frac{G\mathcal{P}_2}{L}
\leq
\mathbb{P}\left\{\mathcal{E}_B \vert M=2\right\}
\leq
Gp\left(1-\eta\right)\mathcal{P}_2.
\end{equation}
We will prove these inequalities in the sequel, where we rigorously derive lower and upper bounds for all ranges of parameters.  

\paragraph{Upper Bound}
To obtain an upper bound, we simply employ the  union bound, i.e. $\mathbb{P}\left\{\mathcal{E}_B\vert S_{\eta},M=2\right\}\leq S_{\eta}\mathbb{P}\left\{\mathcal{E}^{\left(1\right)}_B\vert M=2\right\}$ with $S_{\eta}$ indicating the number of discriminating SNPs. The obtained upper bound can be formulated as follows:
\begin{equation}
\mathbb{P}\left\{\mathcal{E}_B \vert M=2\right\}
=
\mathbb{E}_{S_{\eta}}\left\{\mathbb{P}\left\{\mathcal{E}_B \vert S_{\eta},M=2\right\}\right\}
\leq
Gp\left(1-\eta\right)\mathcal{P}_2.
\end{equation}


\paragraph{Lower Bound}
To obtain a lower bound on $\pbb{\mathcal{E}_B\vert M=2}$, we obtain an upper bound on $\pbb{\mathcal{\bar E}_B\vert M=2}$. Let $\mathcal{C}_1$ represent the event that there are less than two discriminating SNPs in the whole genome. Clearly, bridging condition is satisfied under $\mathcal{C}_1$. Moreover, let $\mathcal{C}_2$ represent the event that for any fragment of length $L$ between the first and the last discriminating SNPs, there exists one read starting within the fragment and covering at least one discriminating SNP.  If $\mathcal{C}_2$  does not happen, then bridging condition fails as the flow of information does not go through discriminating SNPs and the assembly stops at the fragment violating the condition. Therefore, 
\begin{align}
\pbb{\mathcal{\bar E}_B\vert M=2} &  \leq \pbb{\mathcal{C}_1 \cup \mathcal{C}_2} \nonumber \\
& = \pbb{\mathcal{C}_1} +\pbb{\mathcal{C}_2 | \mathcal{\bar C}_1} \left(1- \pbb{\mathcal{C}_1}\right).
\label{eq:twoIndLowerBound1}
\end{align}  
It can be easily verified that 
\begin{equation}
\pbb{\mathcal{C}_1} =\left(1+ Gp\left(1-\eta\right) \right) e^{-Gp\left(1-\eta\right)}. 
\label{eq:twoIndLowerBound2}
\end{equation}

To obtain an upper bound on $\pbb{\mathcal{C}_2 | \mathcal{\bar C}_1}$, we note that if the distance between the first and the last discriminating SNPs, denoted by $L_R$, is partitioned into $L_R/L$ non-overlapping segments of length $L$, then it is necessary that for all segments, at least one read starts and then covers at least one discriminating SNP. Clearly, focusing on such non-overlapping segments provide an upper bound on $\pbb{\mathcal{C}_2 | \mathcal{\bar C}_1}$. Moreover, since the segments are non-overlapping, the corresponding events are independent. Let $\mathcal{C}_2^j$ denote the event corresponding to the $j$th segment. Then,
\begin{align}
\pbb{\mathcal{C}_2 | \mathcal{\bar C}_1} & = \mathbb{E}_{L_R} \left\{ \pbb{\mathcal{C}_2 | \mathcal{\bar C}_1, L_R}\right\} \nonumber \\
& \leq \mathbb{E}_{L_R}  \left\{ \pbb{\bigcap_{j=1}^{L_R/L}\mathcal{C}_2^j \big | \mathcal{\bar C}_1 , L_R} \right\}\nonumber\\
& =  \mathbb{E}_{L_R}  \left\{ \left( \pbb{\mathcal{C}_2^1 | \mathcal{\bar C}_1 , L_R} \right)^{\frac{L_R}{L}} \right\}. \label{eq:lower_bound_1}
\end{align}

We further obtain an upper bound on $\pbb{\mathcal{C}_2^1 | \mathcal{\bar C}_1 , L_R}$ using the following arguments. If there are $n$ total arrivals within this interval and $k$ of them belong to discriminating SNPs and the rest belong to DNA reads from one of the two individuals, then the event does not happen if all SNPs arrive  before all  DNA reads. First, note that the probability of observing $k$ discriminating SNPs and $n-k$ reads within the segment can be obtained as
\begin{equation}
\frac{
\left(p\left(1-\eta\right)\right)^k\left(2\lambda\right)^{n-k}
}{k!\left(n-k\right)!}
e^{-\left(2\lambda+p\left(1-\eta\right)\right)L}.
\nonumber
\end{equation}
Averaging over all $n$ and $k$ yields
\begin{align}
\pbb{\mathcal{C}_2^1 | \mathcal{\bar C}_1 , L_R} 
& = 1- 
\sum_{n=0}^{\infty}\sum_{k=0}^{n}
\frac{1}{\binom{n}{k}}
\frac{
\left(p\left(1-\eta\right)\right)^k\left(2\lambda\right)^{n-k}
}{k!\left(n-k\right)!}
e^{-\left(2\lambda+p\left(1-\eta\right)\right)L}
\end{align}
A simple calculation reveals that
\begin{align}
\pbb{\mathcal{C}_2^1 | \mathcal{\bar C}_1 , L_R} 
& = 1- \frac{2\lambda e^{-p\left(1-\eta\right)L} - p\left(1-\eta\right)e^{-2\lambda L}}{2\lambda-p\left(1-\eta\right)} \nonumber \\
& = 1- \mathcal{P}_2.
\end{align}
Substituting into \eqref{eq:lower_bound_1} yields
\begin{align}
\pbb{\mathcal{C}_2 | \mathcal{\bar C}_1} 
& \leq  \mathbb{E}_{L_R}  \left\{ \left( 1- \mathcal{P}_2\right)^{\frac{L_R}{L}} \right\}. 
\end{align}
The distribution of $L_R$ is needed to compute the expectation. It is easy to show that
\begin{equation}
L_R~\sim~
\frac{
p^2\left(1-\eta\right)^2
\left(G-L_R\right)e^{-p\left(1-\eta\right)\left(G-L_R\right)}
}{
1-\left(1+Gp\left(1-\eta\right)\right)e^{-Gp\left(1-\eta\right)}
}
\quad,\quad 0\leq L_R\leq G.
\end{equation}
The final formulation for upper bound can be written as follows:
\begin{align}
\pbb{\mathcal{E}_B\vert M=2}
&\ge
1-\pbb{\mathcal{C}_1}-\pbb{\mathcal{C}_2 | \mathcal{\bar C}_1} \left(1- \pbb{\mathcal{C}_1}\right)
\nonumber \\
&\ge
1-\left(1+Gp\left(1-\eta\right)\right)e^{-Gp\left(1-\eta\right)}-
\nonumber \\
&~~\frac{
\left(1-\mathcal{P}_2\right)^{\frac{G}{L}}
-\left(1+\frac{G}{L}\left(Lp\left(1-\eta\right)+\log\left(1-\mathcal{P}_2\right)\right)\right)
e^{-Gp\left(1-\eta\right)}
}{
\left(1+\frac{\log\left(1-\mathcal{P}_2\right)}{Lp\left(1-\eta\right)}\right)^2
}
\nonumber \\
&\triangleq \Lambda\left(\mathcal{P}_2\right).
\label{eq:psiDef}
\end{align}

In the asymptotic regime where $G\gg L$ and $\min\left\{p\left(1-\eta\right),2\lambda\right\}L\gg1$ (which is the case in all practical situations), such lower bound can be approximated by the following formulation:
\begin{equation}
\mathbb{P}\left\{\mathcal{E}_B\vert M=2\right\}
\ge	~
1-\exp\left(-\frac{G\mathcal{P}_2}{L}\right)
\simeq
\frac{G\mathcal{P}_2}{L}.
\end{equation}
Clearly, both lower and upper bounds for bridging error probability have the same exponent of decay with respect to $L$. This property indicates a sharp phase transition for $\pbb{\mathcal{E}_B\vert M=2}$ as $L$ is increased.

\subsubsection{Case of $M$ Individuals}

The upper and lower bounds obtained in the two individual case can be readily generalized to $M>2$ individuals with a tight asymptotic performance. This result is presented in the following theorem.
\begin{thm2}
\label{thm:MTheorem}
The probability of error event in bridging condition for $M$ individuals, $\pbb{\mathcal{E}_B}$, can be bounded both from above and below as:
\begin{equation}
\Lambda\left(\Delta\right)
\leq
\mathbb{P}\left\{\mathcal{E}_B\right\}
\leq
\binom{M}{2}Gp\left(1-\eta\right)\mathcal{P}_2,
\nonumber
\end{equation}
where $\Delta$ is defined as follows:
\begin{equation}
\Delta\triangleq
\sum_{m=2}^{M}\left(-1\right)^m\binom{M}{m}\mathcal{P}_m~,~
\mathcal{P}_m \triangleq \left\{
\begin{array}{cc}
\frac{m\lambda e^{-p\left(1-\eta\right)L} - p\left(1-\eta\right)e^{-m\lambda L}}{m\lambda-p\left(1-\eta\right)}
&,~  \lambda\neq p\left(\frac{1-\eta}{m}\right)
\\[3mm]
\left(1+p\left(1-\eta\right)L\right)e^{-p\left(1-\eta\right)L} &,~ \lambda=p\left(\frac{1-\eta}{m}\right)
\end{array}
\right.
\label{eq:DeltaDef}
\end{equation}
and $\Lambda\left(\cdot\right)$ is defined as in \eqref{eq:psiDef}.
\end{thm2}

\begin{proof}
The upper bound has been directly derived from union bound in which the factor of $\binom{M}{2}$ corresponds to  all pairwise comparisons between every two individuals. Each of such comparisons has an error probability of at most $\mathcal{P}_2$ which justifies one of the inequalities.

For the lower bound, we follow similar arguments to those in the case of two individuals. Let $L_R$ denote the distance between the first and last SNPs in genome. We divide this region into $L_R/L$ non-overlapping segments of length $L$. In order to have a unique assembly, it is necessary that at least $M-1$ DNA reads from $M-1$ distinct individuals arrive in each segment and cover at least one discriminating SNP. If this condition is violated, for at least two individuals, the two sides of the segment are informatically disconnected and therefore  unique assembly is impossible. The probability of error for such event in a single segment can be derived as follows. Assume that there are $k$ discriminating SNPs and $n_m$ DNA reads from the $m$th individual in a particular segment, where $m=1,2,\ldots,M$. Obviously, $k$ and each $n_m$ are random variables with corresponding {\it {Poisson}} distributions. Also, let $w$ denote the number of all possible permutations of SNPs and DNA reads arrivals:
\begin{equation}
w=\frac{\left(n_1+\cdots+n_M+k\right)!}{k!n_1!\cdots n_M!}.
\end{equation}	
An error happens whenever for at least two individuals all DNA reads arrive after the last discriminating SNP. We first consider the probability of this event for two particular individuals, namely the $i$th and the $j$th ones. The number of all possible permutations that correspond to such scenario can be calculated by first ordering the $k$ discriminating SNPs and then ordering $n_i+n_j$ DNA reads associated with the $i$th and $j$th individuals. The rest of DNA reads which belong to other individuals can then be arbitrarily distributed among these arrivals. Consequently, the total number of {\it {bad}} permutations is $w/\binom{n_i+n_j+k}	{k}$, which implies that the probability of occurring an error with respect to only the $i$th and $j$th individuals is ${\binom{n_i+n_j+k}{k}}^{-1}$. One can readily show that when there are $c$ individuals with indices $i_1,i_2,\ldots,i_c$, this probability generalizes to ${\binom{n_{i_1}+\cdots+n_{i_c}+k}{k}}^{-1}$. By using the inclusion-exclusion principle, it can be shown that the probability of occurring an error event in a single segment $\Delta$ can be obtained via the following formulation:
\begin{align}
\Delta
\triangleq ~
e^{-\left(p\left(1-\eta\right)+M\lambda\right)L}
\sum_{k=0}^{\infty}
\sum_{n_1=0}^{\infty}
\ldots
\sum_{n_M=0}^{\infty}&
\frac{\left(p\left(1-\eta\right)L\right)^k}{k!}
\frac{\left(\lambda L\right)^{n_1+\cdots+n_M}}{n_1!\cdots n_M!}
\left(
\sum_{i<j}{\binom{n_i+n_j+k}{k}}^{-1}-
\right. \nonumber \\ & \left.
\cdots
+\left(-1\right)^M
{\binom{n_1+\cdots+n_M+k}{k}}^{-1}
\right),
\end{align}
which can be simplified into the formulation given in \eqref{eq:DeltaDef}. Again, pursuing the same path we have followed in the case of two individuals one can show that:
\begin{equation}
\pbb{\mathcal{E}_B}\ge
\left(
1-\left(1+Gp\left(1-\eta\right)\right)e^{-Gp\left(1-\eta\right)}
\right)\left(1-
\mathbb{E}_{L_R}\left\{\left(1-\Delta\right)^{\frac{L_R}{L}}\right\}
\right)
=\Lambda\left(\Delta\right),
\end{equation}
which completes the proof.
\end{proof}

Next, we address asymptotic analysis of the attained bounds. For the sake of simplicity, let us assume $m\lambda\neq p\left(1-\eta\right)$ for all $m=2,3,\ldots$. The special cases where the equality holds for some $m$ lead to very similar analyses as the following arguments. Based on this assumption, it can be shown that $\mathcal{P}_m,m=2,3,\ldots$ satisfies the following inequalities:
\begin{equation}
e^{-\min\left\{p\left(1-\eta\right),m\lambda\right\}L}
\leq
\mathcal{P}_m
\leq
\frac{e^{-\min\left\{p\left(1-\eta\right),m\lambda\right\}L}}
{1-\frac{\min\left\{p\left(1-\eta\right),m\lambda\right\}}{\max\left\{p\left(1-\eta\right),m\lambda\right\}}}.
\end{equation}
Consequently, when $G\gg L$ and $\min\left\{p\left(1-\eta\right),2\lambda\right\}L\gg1$, bridging error probability $\pbb{\mathcal{E}_B}$ can be simply bounded as:
\begin{align}
\frac{G}{L}
\sum_{m=2}^{M}\frac{\left(-1\right)^m\binom{M}{m}}{1-\frac{p\left(1-\eta\right)}{m\lambda}}
e^{-\min\left\{m\lambda,p\left(1-\eta\right)\right\}L}
&\leq
\mathbb{P}\left\{\mathcal{E}_B\right\}
\nonumber \\
&\leq
\frac{GM^2 p\left(1-\eta\right)}{2\left(1-\frac{\min\left\{2\lambda,p\left(1-\eta\right)\right\}}{\max\left\{2\lambda,p\left(1-\eta\right)\right\}}\right)}
e^{-\min\left\{2\lambda,p\left(1-\eta\right)\right\}L}.
\end{align}
And if $\lambda\gg p\left(1-\eta\right)$, then the lower bound can be simply approximated by 
\begin{equation}
\frac{G\left(M-1\right)}{L}e^{-p\left(1-\eta\right)L}.
\end{equation}

It is evident that the increase in both upper and lower bounds of $\pbb{\mathcal{E}_B}$ are polynomial with respect to $G$ and $M$, while the decrease in error remains exponential with respect to $L$. The exponent of decay for the upper bound is $\min\left\{2\lambda,p\left(1-\eta\right)\right\}$. This result implies that if $\lambda> p\left(\frac{1-\eta}{2}\right)$, then any further increase in sequencing depth does not lead to a significantly better asymptotic behavior.

\subsection{Assembly Error: Asymptotic Analysis}

In this subsection, we combine lower and upper bounds associated with SNP coverage and bridging conditions to obtain asymptotically tight bounds on the assembly error probability $\pbb{\mathcal{E}}$. Based on \eqref{eq:totErrorBound}, when $\min\left\{\lambda,p\left(1-\eta\right)\right\}L\gg1$ and $G\gg L$ a global upper bound on the assembly error probability can be obtained as follows:
\begin{align}
\mathbb{P}\left\{\mathcal{E}\right\}\leq&~
\mathbb{P}\left\{\mathcal{E}_{SC}\right\}+
\mathbb{P}\left\{\mathcal{E}_B\right\}
\nonumber \\
\leq&~
GM\left(
\frac{e^{-\lambda L}}{\frac{1}{p}+\frac{1}{\lambda}}
+
\frac{Mp\left(1-\eta\right)e^{-\min\left\{2\lambda,p\left(1-\eta\right)\right\} L}}
{2\left(1-\frac{\min\left\{2\lambda,p\left(1-\eta\right)\right\}}{\max\left\{2\lambda,p\left(1-\eta\right)\right\}}\right)}
\right).
\end{align}
The lower bound can also be written as:
\begin{align}
\mathbb{P}\left\{\mathcal{E}\right\}\ge&~
\max\left\{
\mathbb{P}\left\{\mathcal{E}_{SC}\right\},
\mathbb{P}\left\{\mathcal{E}_B\right\}
\right\}
\nonumber \\
\ge&~
G\max\left\{
\frac{e^{-\lambda L}}{\frac{1}{p}+\frac{1}{\lambda}}
\max\left\{
1,
\frac{
M\left(1+\frac{p}{\lambda}\right)^{\frac{-\lambda}{p}}
}
{\lambda L + 
\frac{\lambda}{p}\log\left(
1+\frac{p}{\lambda}
\right)
}
\right\}
,
\frac{1}{L}
\sum_{m=2}^{M}\frac{\left(-1\right)^m\binom{M}{m}}{1-\frac{p\left(1-\eta\right)}{m\lambda}}
e^{-\min\left\{m\lambda,p\left(1-\eta\right)\right\}L}
\right\}.
\end{align}
If one aims at increasing the sequencing depth such that $2\lambda\gg p\left(1-\eta\right)$, then the above asymptotic lower and upper bounds can be further simplified resulting in the following bounds:
\begin{equation}
\frac{G\left(M-1\right)}{L}
e^{-p\left(1-\eta\right)L}
\leq
\pbb{\mathcal{E}}
\leq
\frac{1}{2}GM^2p\left(1-\eta\right)e^{-p\left(1-\eta\right)L}.
\end{equation}
It should be reminded that exact non-asymptotic lower and upper bounds can be found in Theorem \ref{thm:MTheorem}, \eqref{eq:exactLower1} and \eqref{eq:SCLowerBound}.

In this section, we have derived asymptotically tight lower and upper bounds on $\pbb{\mathcal{E}}$ when DNA reads are tagless and noiseless. However, all practical sequencing machines produce DNA reads that are contaminated with sequencing noise which alters the sequenced nucleotides by a particular error rate. Next section deals with theoretical bounds on the assembly error rate in noisy scenarios.

\section{Assembly Analysis: Noisy Fragments}
\label{sec:noisy}

In this section, we consider the case of noisy DNA reads where biallelic SNP values are randomly altered with an error probability of $\epsilon$. Moreover, we have assumed binary SNP values, although generalization to more than two alleles is straightforward. Surprisingly, we have shown that for any $\epsilon < 0.5$ and for sufficiently large number of reads, unique and correct genome assembly from noisy reads is possible. Note that $\epsilon = 0.5$ corresponds to the case where reads do not carry any information related to SNPs and assembly is impossible. 

The following theorem provides sufficient conditions for genome assembly in the noisy case. The conditions are so stringent that a simple greedy algorithm such as the one presented in the noiseless case can reconstruct all the genomes correctly and uniquely. 

\begin{thm2}[Sufficient Conditions for Unique Genome Assembly in Noisy Case]
\label{thm:MainNoisyTheorem}
Assume genome is divided into a number of overlapping segments of length $D$, where each two consecutive segments have an overlapping region of length $D-d$ ($d<D$). Also assume the following conditions hold:
\begin{itemize}
\item
Discrimination Condition (Disc): In each overlapping region between two consecutive segments, every two individuals are distinguishable based on their true genomic content,
\item
Denoising Condition (Den): In each segment, one can phase all the individuals based on reads covering the segment.
\end{itemize}
Then, correct and unique genome assembly is possible for all individuals.
\end{thm2}

\begin{proof}
The proof is straightforward. If the above-mentioned conditions hold, then the conditions in Theorem \ref{thm:MainTheorem} also hold and genome assembly becomes possible by the greedy algorithm.
\end{proof}

Let us denote the error event in correct and  unique assembly of all genomes in the noisy regime by $\mathcal{E}_{N}$. From the fact that the genome can be divided into $G/d$ overlapping segments of length $D$, one can write
\begin{align}
\pbb{\mathcal{E}_{N}} &  \leq \frac{G}{d} \pbb{\mathcal{E}^{\left(1\right)}_{N}},
\end{align}
where $\mathcal{E}^{\left(1\right)}_{N}$ is the event that in an interval of length $D$ one of the sufficient conditions does not hold. Let us denote $\mathcal{E}_{Disc}$ and $\mathcal{E}_{Den}$ as the error events corresponding to discrimination and denoising conditions in a single segment, respectively. Then,  from $\mathcal{E}^{\left(1\right)}_{N} =\mathcal{E}_{Disc} \cup \mathcal{E}_{Den} $ we have:
\begin{align}
\pbb{\mathcal{E}_N}
&\stackrel{\text{(a)}}{\leq}
 \frac{G}{d}
\sum_{\kappa =0}^\infty
\pbb{\mathcal{E}_{Disc}\cup \mathcal{E}_{Den} | \kappa }\pbb{\kappa} 
\nonumber \\
&\stackrel{\text{(b)}}{\leq}
 \frac{G}{d}
\sum_{\kappa =0}^\infty
\left(
\pbb{\mathcal{E}_{Disc}| \kappa } + \pbb{ \mathcal{E}_{Den} |\mathcal{\bar E}_{Disc}, \kappa }
\right)
\pbb{\kappa} 
\nonumber \\
&\leq
 \frac{G}{d}
\pbb{\mathcal{E}_{Disc}} +  
\frac{G}{d}
\sum_{\kappa =0}^\infty
 \pbb{ \mathcal{E}_{Den} |\mathcal{\bar E}_{Disc}, \kappa }\pbb{\kappa},
\end{align}
where (a) is obtained by conditioning on $\kappa$ which is the number of SNPs within a given segment of length $D$, and (b) is obtained by upper bounding $1-\pbb{\mathcal{\bar E}_{Disc} |\kappa}$ by $1$. 

The event $\mathcal{E}_{Disc}$ is independent of reads. Therefore, one can simply attain an upper bound on $\pbb{\mathcal{E}_{Disc}}$ using the fact that the number of SNPs in an overlap region of length $D-d$ has a {\it {Poisson}} distribution. Therefore,
\begin{align}
\pbb{\mathcal{E}_{Disc}}
&\stackrel{\text{(a)}}{\leq}
\sum_{n=0}^{\infty}
e^{-p\left(D-d\right)}\frac{\left[p\left(D-d\right)\right]^n}{n!}
\left(
1-
\left(1-\eta^n\right)^{\binom{M}{2}}
\right)
\nonumber \\
&=
\sum_{m=1}^{\binom{M}{2}}
\binom{\binom{M}{2}}{m}\left(-1\right)^{m-1}
e^{-p\left(D-d\right)\left(1-\eta^m\right)},
\end{align}
where (a) follows from that fact that $1-\left(1-\eta^n\right)^{\binom{M}{2}}$ is an upper bound on the probability that every two individuals are distinguishable based on a set of $n$ observed SNPs.

The event $\{\mathcal{E}_{Den} |\mathcal{\bar E}_{Disc}, \kappa \}$ depends on the reads as well as the algorithm used for denoising. We propose two algorithms in this regard. The first one, which leads to the optimal solution, is based on maximum likelihood and will be presented in Section \ref{sec:perfect_denoising}. This algorithm yields the best performance at the expense of a prohibitive computational complexity. In Section \ref{sec:spectral_denoising}, we present an algorithm which is motivated by random graph theory as an alternative approach demonstrating high performance during the experiments while retaining similar interesting asymptotic properties.

\subsection{Optimal Denoising Algorithm : Maximum Likelihood} \label{sec:perfect_denoising}

Evidently, decision making based on Maximum Likelihood (ML) is the optimal denoising algorithm which searches among all possible choices of the hidden SNP sets and chooses the one with the maximum probability of observation conditioned on the SNP sets.

Theorem \ref{thm:NoisyTheorem} guarantees that given the genomic contents of each two individuals are distinguishable in each segment, then denoising error associated with ML goes to zero when the sequencing depth goes to infinity. Moreover, it  has been proved that decoding error exponentially decays as the number of observations per segment is increased.

\begin{thm2}
\label{thm:NoisyTheorem}
Assume $M$ distinct and hidden SNP sets each having $\kappa$ SNPs. Then, there exist non-negative real functions $\mathcal{D}_1\left(\epsilon,\kappa\right),\ldots,\mathcal{D}_{M\kappa}\left(\epsilon,\kappa\right)$ such that
\begin{align}
\pbb{\mathcal{E}_{Den} |\mathcal{\bar E}_{Disc}, \kappa}
&\leq
\sum_{i=1}^{M\kappa}
\binom{M\kappa}{i}
e^{-\lambda M\left(L-D\right)\left(1-
e^{-\mathcal{D}_i\left(\kappa,\epsilon\right)}\right)}.
\end{align}
Also, we have $\mathcal{D}_i\left(\epsilon,\kappa\right)>0$ for all $i$, if and only if $0\leq\epsilon<0.5$.
\end{thm2}

\begin{proof}

\begin{figure}[t]
\centering
	\includegraphics[trim=2cm 8cm 5cm 4cm,clip,width=0.9\textwidth]{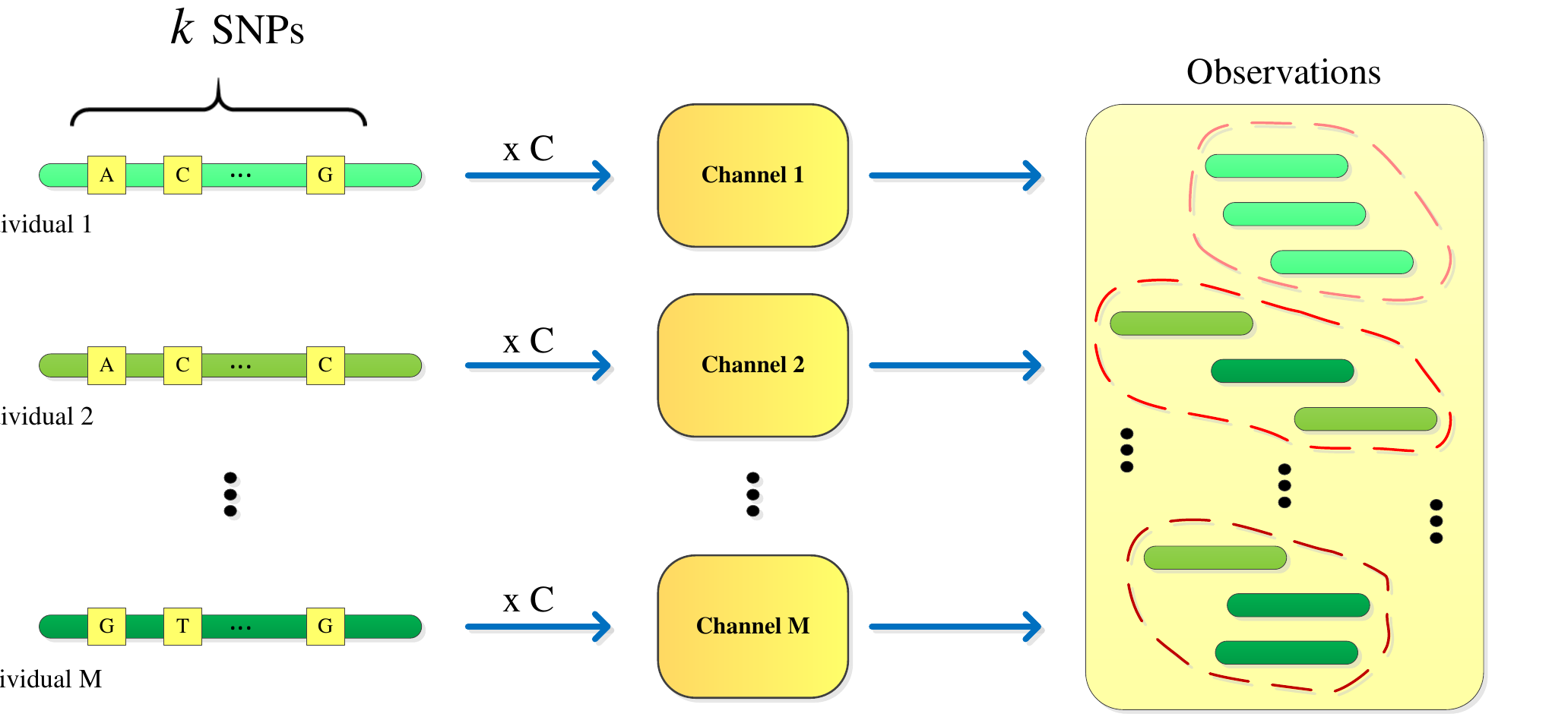}
	\caption{The problem model for block-wise noise removal. True underlying SNP sequences are transmitted through identical yet independent binary noisy channels for several times to form a pool of noisy observations.}
	\label{fig:NoisyChannels}
\end{figure}

The problem model for Theorem \ref{thm:NoisyTheorem} is depicted in Fig. \ref{fig:NoisyChannels}. Let us denote $\Phi_{\kappa}=\left\{\phi_1,\phi_2,\ldots,\phi_{2^\kappa}\right\},\phi_i\in\left\lbrace0,1\right\rbrace^{\kappa}$ as the set of all possible $2^\kappa$ sequences of a binary SNP segment with length $\kappa$. Also, let $\Psi_{\kappa}$ represent the set of all $\binom{2^{\kappa}}{M}$ possible subsets of $\Phi_{\kappa}$ that have a cardinality of $M$. In this regard, the true set of SNP sequences underlying the segment will be denoted by $\psi_T\in\Psi_{\kappa}$. It is straightforward that for any $\phi_i\in\Phi_{\kappa}$, the probability of observing $\phi_i$ in the observation pool can be written as:
\begin{equation}
\pbb{\phi_i\vert \psi_T}=\frac{\left(1-\epsilon\right)^\kappa}{M}\sum_{j=1}^{M}\left(\frac{\epsilon}{1-\epsilon}\right)^{\rho_{ij}}~,~i=1,2,\ldots,2^{\kappa}
\end{equation} 
where $\rho_{ij}$ denotes the {\it {Hamming}} distance between $\phi_i$ and the $j$th sequence in $\psi_T$.

We attempt to find $\psi_T\in\Psi_{\kappa}$ through a series of pairwise hypothesis testings. In other words, maximum likelihood decoding implies the following maximization problem:
\begin{equation}
\psi^{\left(\text{ML}\right)}_T=\argmax_{\psi\in\Psi_{\kappa}}~\prod_{i=1}^{n}\pbb{R_i\vert \psi},
\end{equation}
where $\boldsymbol{R}=\left\{R_1,R_2,\ldots,R_n\right\}$ denotes the set of $n$ observations.

An error in ML decoding occurs whenever for at least one $\psi\in\Psi_{\kappa}$ and $\psi\neq\psi_T$, we have $\pbb{\boldsymbol{R}\vert \psi_T}\leq\pbb{\boldsymbol{R}\vert \psi}$. As a result, one can use the union bound to upper bound the decoding error probability as follows. For all $i=1,2,\ldots,M\kappa$, assume $\Psi_{\kappa}^{\left(i\right)}$ a subset of $\Psi_{\kappa}$ such that every member of $\Psi_{\kappa}^{\left(i\right)}$ has a {\it {Hamming}} distance of $i$ from $\psi_T$. In other words, through each alternation of $i$ out of $M\kappa$ SNPs in $\psi_T$ we obtain  a member of $\Psi_{\kappa}^{\left(i\right)}$. Evidently, one can show that $\left\vert\Psi_{\kappa}^{\left(i\right)}\right\vert\leq\binom{M\kappa}{i}$. In this regard, the error probability in decoding a segment consisting of $\kappa$ SNPs with $n$ observations, $\pbb{\mathcal{E}_{Den} |\mathcal{\bar E}_{Disc}, \kappa, n}$, can be bounded as
\begin{align}
\pbb{\mathcal{E}_{Den} |\mathcal{\bar E}_{Disc}, \kappa, n}
\leq
\sum_{i=1}^{M\kappa}\binom{M\kappa}{i}\max_{\psi_T\in\Psi_{\kappa},\atop\psi\in\Psi_{\kappa}^{\left(i\right)}}
\pbb{\pbb{\boldsymbol{R}\vert\psi_T} \leq \pbb{\boldsymbol{R}\vert\psi}}.
\end{align}
Also, we have
\begin{equation}
\pbb{
\pbb{\boldsymbol{R}\vert \psi_T}\leq
\pbb{\boldsymbol{R}\vert \psi}
}
=
\pbb{
\sum_{i=1}^{n}\log\left(\frac{\pbb{R_i\vert \psi_T}}
{\pbb{R_i\vert \psi}}\right)\leq 0
}.
\label{eq:MLUnionPartition}
\end{equation}
The inequality in \eqref{eq:MLUnionPartition} is based on the deviation of a sum of i.i.d. random variables around its mean. By using  Chernoff bound, one can attain the following inequality which upper bounds the denoising error probability:
\begin{align}
\pbb{\mathcal{E}_{Den} |\mathcal{\bar E}_{Disc}, \kappa, n }
&\leq
\max_{\psi_T\in\Psi_{\kappa}}
\sum_{i=1}^{M\kappa}
\binom{M\kappa}{i}
\exp\left(
-n
\min_{\psi\in \Psi^{\left(i\right)}_{\kappa}}
\sup_{s\in\mathbb{R}}
\log
\left[
\left(
\sum_{\phi\in\Phi_{\kappa}}
\frac{
\left(\pbb{\phi\vert \psi_T}\right)^{s+1}
}{
\left(\pbb{\phi\vert \psi}\right)^{s}
}
\right)^{-1}
\right]
\right)
\nonumber \\
&\triangleq
\sum_{i=1}^{M\kappa}
\binom{M\kappa}{i}
e^{-n\mathcal{D}_i\left(\epsilon,\kappa\right)}.
\label{eq:FinalMLBound}
\end{align}
It is easy to show that the minimal (worst) exponent of error corresponds to $\mathcal{D}_1\left(\kappa,\epsilon\right)$, since it considers all two distinct members of $\Psi_{\kappa}$ that have the minimum {\it {Hamming}} distance. Furthermore, we need to show that the exponent of error is positive when $\epsilon<0.5$. Lemma \ref{lemma:positiveExponent} and \ref{lemma:uniqueness} guarantee the positivity of error exponent under the above-mentioned conditions.
\begin{lemma}
Assume $P$ and $Q$ are two arbitrary distributions over a finite $\sigma$-filed. Then, we have:
\begin{equation}
\sup_{s\in\mathbb{R}}\log\left[\left(\sum_{i}\frac{P^{s+1}_i}{Q^s_i}\right)^{-1}\right]=
\log\left[\left(\sum_{i}\sqrt{P_i Q_i}\right)^{-1}\right]
\ge 0.
\label{eq:Lemma1eq}
\end{equation}
In addition, the equality holds only for $P=Q$.
\label{lemma:positiveExponent}
\end{lemma}
The proof of Lemma \ref{lemma:positiveExponent} is given in Appendix \ref{app:Dmin}. According to Lemma \ref{lemma:positiveExponent}, the error exponents in \eqref{eq:FinalMLBound} are strictly positive given that each two different members of $\Psi_{\kappa}$ result into different statistical distributions over $\Phi_{\kappa}$. Mathematically speaking:
\begin{lemma}
Assume $\psi,\psi'\in\Psi_{\kappa}$ and $\psi\neq\psi'$, then $\pbb{\phi\vert\psi}\neq\pbb{\phi\vert\psi'}$.
\label{lemma:uniqueness}
\end{lemma}
The proof of Lemma \ref{lemma:uniqueness} is discussed in Appendix \ref{app:Dmin}. We have also derived analytical formulations for the main exponent, i.e. $\mathcal{D}_1\left(\epsilon,\kappa\right)$ for special cases of $M=2$ and $M=3$. Surprisingly, it can be shown that the main error exponent does not depend on $\kappa$ and is only a function of $M$ and $\epsilon$.
\begin{lemma}
The main error exponent $\mathcal{D}_1$ for the special cases of $M=2$ and $3$ has the following analytical formulations:
\begin{equation}
e^{-\mathcal{D}_1}=\left\{\begin{array}{lc}
\frac{1}{2}\left(1-2\sqrt{\epsilon\left(1-\epsilon\right)}\right) & M=2
\\
\frac{2}{3}\sqrt{1+\epsilon\left(1-\epsilon\right)}
\left(
\sqrt{2\epsilon\left(1-\epsilon\right)+\epsilon^2}
+
\sqrt{2\epsilon\left(1-\epsilon\right)+\left(1-\epsilon\right)^2}
\right)
&
M=3
\end{array}\right..
\end{equation}
\label{lemma:sampleDmin}
\end{lemma}
Again, the interested reader can find the proof of Lemma \ref{lemma:sampleDmin} in Appendix \ref{app:Dmin}. The procedure used for obtaining analytic formulation of $\mathcal{D}_1$ can be used for cases where $M>3$. Moreover, it is easy to show that for all $M>2$, we have:
\begin{equation}
\lim_{\epsilon\rightarrow 0}\mathcal{D}_1\left(\epsilon\right)=\log\left(1+\frac{1}{M-1}\right),
\end{equation}
and $\mathcal{D}_1\left(\epsilon\right)$ is a decreasing function of $\epsilon$. According to Lemma \ref{lemma:positiveExponent}, when $\epsilon=0.5$ then $\mathcal{D}_1=0$.

Inequality \eqref{eq:FinalMLBound} should be marginalized with respect to $n$ in order to attain a formulation for $\pbb{\mathcal{E}_{Den} |\mathcal{\bar E}_{Disc}, \kappa}$. $n$ indicates the number of DNA reads that cover a segment of length $D$, which has a {\it {Poisson}} distribution with parameter $\lambda M\left(L-D\right)$:
\begin{align}
\pbb{\mathcal{E}_{Den} |\mathcal{\bar E}_{Disc}, \kappa}
&\leq
\sum_{i=1}^{M\kappa}
\binom{M\kappa}{i}
\sum_{n=0}^{\infty}
e^{-\lambda M\left(L-D\right)}\frac{\left(\lambda M\left(L-D\right)\right)^n}{n!}
e^{-n\mathcal{D}_i\left(\epsilon,\kappa\right)}
\nonumber \\
&=
\sum_{i=1}^{M\kappa}
\binom{M\kappa}{i}
e^{-\lambda M\left(L-D\right)\left(1-
e^{-\mathcal{D}_i\left(\epsilon,\kappa\right)}\right)}.
\end{align}
This completes the proof.
\end{proof}

\subsection{Spectral Denoising Algorithm}\label{sec:spectral_denoising}

Although the algorithmic approach described in Theorem \ref{thm:NoisyTheorem} reaches the optimal result with respect to the sufficient conditions described in Theorem \ref{thm:MainNoisyTheorem}, however, it is NP-complete with respect to $\kappa$ as it requires an exhaustive search among all $\binom{2^{\kappa}}{M}$ members of $\Psi_{\kappa}$. This amount of complexity is not affordable in real-world applications. Motivated by community recovery techniques from random graphs, we propose an alternative denoising framework which requires much less computational burden for moderate values of $L$ and $\lambda$ and leads to a very good performance in practice.

The whole problem of SNP block denoising can be viewed from a community detection perspective. In this regard, each noisy observation of a SNP block denotes a node in a weightless and undirected random graph $\mathcal{G}=\left(V,E\right)$, where $V=\left\{R_1,\cdots,R_n\right\}$ represents the vertices of the graph. The set of edges $E$ is obtained from the adjacency matrix of $\mathcal{G}$ denoted by $\mathcal{A}$ which is constructed as follows.  

Let ${\boldsymbol{X}}$ denote a $n\times\kappa$ matrix corresponding to $n$ noisy observations of $\kappa$ SNPs. Without loss of generality, we have assumed ${\boldsymbol{X}}\in\left\{-1,1\right\}^{n\times \kappa}$, where $-1$ denotes the presence of major allele. Let us define the sample cross-correlation matrix as
\begin{equation}
\boldsymbol{C}\triangleq \frac{1}{\kappa}{\boldsymbol{X}}{\boldsymbol{X}}^T.
\end{equation}
Then, we define:
\begin{gather}
{\mathcal{A}}_{i,j}=\left\lbrace\begin{array}{lc}
1 & C_{i,j} \ge \tau_c
\\[2mm]
0 & {\text {o.w.}}
\end{array}\right.
\nonumber \\
i,j=1,2,\ldots,n,
\label{eq:GraphConstructFormula}
\end{gather}
where $\tau_c$ is a predefined threshold.  

Let us picture the SNP blocks belonging to an individual as a community in the graph $\mathcal{G}$. Therefore, the ultimate goal is to partition $V$ into $M$ communities, $\left\lbrace V_1,V_2,\ldots,V_M\right\rbrace$, where $V_m$ represents the set of nodes (equivalently SNP blocks) that belong to the $m$th individual. Equivalently, we wish to find $\mathcal{F}:\left\lbrace1,2,\ldots,n \right\rbrace\rightarrow\left\lbrace1,2,\ldots,M\right\rbrace$ which maps each observation to its corresponding individual. When the correct $\mathcal{F}$ is provided, one can denoise the observations via SNP-wise majority voting among those nodes that belong to particular individuals. Spectral techniques in Random Graph Theory have shown promising performances in community detection. In particular, we will make use of  the algorithm presented by \cite{mcsherry2001spectral} which works very well in practice. 

In order to analyze the algorithm, we assume that  $\mathcal{G}$ is generated based  on Erd\"{o}s-R\'{e}nyi graphs with inter (resp. intra) community connection probability of $a$ (resp. $b$). In this case,   \cite{mcsherry2001spectral} has shown that the hidden mapping function $\mathcal{F}$ can be completely recovered with a probability less than
\begin{equation}
1-
n\exp\left(\frac{-n\left(a-b\right)^2}{a c^2  }\right),
\label{eq:McSherryLB}
\end{equation}
where, $c$ is a constant depending on the algorithm employed for community detection. \cite{mcsherry2001spectral} has also shown that for large $n$ the lower bound in \eqref{eq:McSherryLB} is achievable via a spectral technique based on eigen-decomposition of graph adjacency matrix $\mathcal{A}$.

We obtain an upper bound on  $b$, and a lower bound on $a$ as follows. Let $\boldsymbol{\nu}=\left\lbrace\nu_{i,j}\right\rbrace\in\left\{0,1,\ldots,\kappa\right\}^{M\times M}$ denote the sequence difference matrix, where $\nu_{i,j}$ represents the {\it {Hamming}} distance between the $i$th and the $j$th individuals. We define $\nu_{\min}$ as the minimal non-diagonal entry in $\boldsymbol{\nu}$. It is easy to show that 
\begin{align}
\mathbb{E}\left\lbrace
C_{i,j}
\right\rbrace
=&
\left(1-\frac{2\nu_{\mathcal{F}\left(i\right),\mathcal{F}\left(j\right)}}{\kappa}\right)\left(1-2\epsilon\right)^2.
\label{eq:ApproxGaussianPDF}
\end{align}
From $\mathbb{E}\left\lbrace
C_{i,j}
\right\rbrace$, we choose the threshold in edge connectivity as 
$\tau_c=\left(1-2\epsilon\right)^2\left(1-\frac{\nu_{\min}}{\kappa}\right)$. Since we are conditioned on the discrimination among individuals, $\nu_{\min} \in \{1,\cdots,\kappa\}$. The worst case analysis can be carried out by simply setting $\nu_{\min}=1$. On the other hand, the average case analysis is carried out by setting $\nu_{\min}=\kappa (1-\eta)$. The average case analysis is valid for scenarios where $\kappa$ is large. 

Let us denote $e_1$ and $e_2$ as the events in which two nodes that belong to the same community are not connected, and two nodes that belong to different communities become connected, respectively. Since entries of cross-correlation matrix are sum of independent random variables, $a$ and $b$ can be bounded by Chernoff-Hoeffding theorem \cite{schmidt1995chernoff}:
\begin{align}
a
&\geq
1- P_e\triangleq
1- \exp\left(-\frac{\nu^2_{\min}}{\kappa}\left(1-2\epsilon\right)^4\right)
\nonumber \\
b
&\leq
\frac{1}{M\left(M-1\right)}\sum_{i,j=1\atop i\neq j}^{M}
\exp\left(
-\frac{\left(2\nu_{i,j}-\nu_{\min}\right)^2}{\kappa}\left(1-2\epsilon\right)^4
\right)\le P_e.
\end{align}
It is easy to show that if $b \leq 0.5 \leq a$, then the probability of reconstruction in \eqref{eq:McSherryLB} can be bounded from above by 
\begin{equation}
1-
n\exp\left(\frac{-n\left(1-2P_e(\kappa)\right)^2}{ c^2 \left( 1- P_e(\kappa)\right)  }\right).
\end{equation}

In order to correctly cluster the read set in the asymptotic regime, i.e. $n\rightarrow\infty$, it suffices that the following condition holds:
\begin{equation}
P_e(\kappa)=
\exp\left(
-\frac{\nu^2_{\min}}{\kappa}\left(1-2\epsilon\right)^4
\right)
<
\frac{1}{2}
\end{equation}
which simplifies to the following inequality:
\begin{equation}
\epsilon<
\frac{1}{2}-\frac{1}{2}\sqrt[4]{\frac{\kappa\log 2}{\nu^2_{\min}}}
\xrightarrow{\kappa\gg 1}
\frac{1}{2}\left(1-\sqrt[4]{\frac{\log 2}{\kappa\left(1-\eta\right)^2}}\right).
\label{eq:spectralErrorIneq}
\end{equation}
As a result, the maximum affordable error rate in (\ref{eq:spectralErrorIneq}) converges to $\epsilon<0.5$ as long as $\kappa$ is chosen sufficiently large. Given a procedure to perfectly cluster the read set, denoising can be performed by taking SNP-wise majority vote among those observations that belong to the same individual. In this way, one can arbitrarily reduce the error rates if a sufficiently large number of independent noisy reads is provided. For finite $n$, probability of error event due to majority voting, denoted by $e_{mv}$, can be upper bounded via Chernoff bound as follows:
\begin{equation}
\pbb{e_{mv}}\leq
M\kappa \exp\left({\frac{-n}{8\epsilon\left(1-\epsilon\right)}}\right).
\end{equation}
Since the error event in spectral denoising is the union of error in community detection and majority voting, and by assuming that the lower bound in \eqref{eq:McSherryLB} is achievable, one can simply write:
\begin{align}
\pbb{\mathcal{E}_{Den} \vert\mathcal{\bar E}_{Disc}, \kappa}
&\leq
e^{-\lambda M\left(L-D\right)}
\sum_{n=0}^{\infty}
\frac{\left(\lambda M\left(L-D\right)\right)^n}{n!}
\left(
ne^{-n\zeta\left(\kappa\right)}
+
M\kappa \exp\left({\frac{-n}{8\epsilon\left(1-\epsilon\right)}}\right)
\right)
\nonumber \\
&=
\lambda M\left(L-D\right)
e^{-\zeta\left(\kappa\right)}
e^{-\lambda M\left(L-D\right)\left(1-e^{-\zeta\left(\kappa\right)}\right)}
+
M\kappa e^{-\lambda MD\left(1-\exp\left(\frac{-1}{8\epsilon\left(1-\epsilon\right)}\right)\right)},
\end{align}
where $\zeta\left(\kappa\right)$ is defined as follows:
\begin{equation}
\zeta\left(\kappa\right)\triangleq
\frac{\left(1-2P_e\left(\kappa\right)\right)^2}{c^2\left(1-P_e\left(\kappa\right)\right)}.
\label{eq:ZetaDef}
\end{equation}

\subsection{Assembly Error: Asymptotic Analysis}
\label{sec:noisyasymptot}

So far we have derived upper bounds on the error probabilities associated to discrimination and denoising conditions, i.e. $\pbb{\mathcal{E}_{Disc}}$ and $\pbb{\mathcal{E}_{Den}}$, respectively. The derived bounds hold for any $0<d\leq D\leq L$. In this part, we will integrate the obtained results in order to propose appropriate upper bounds on the  assembly error rate. The achieved error upper bounds in addition to assembly regions are depicted for both maximum likelihood and spectral denoising algorithms. Also, we have analyzed  the asymptotic behavior of our bounds under certain circumstances.
\subsubsection{Maximum Likelihood Denoising}

One can combine the attained results for the maximum likelihood (ML) denoising to obtain an upper bound on the assembly error probability in the noisy regime $\pbb{\mathcal{E}_N}$ as follows:
\begin{align}
\pbb{\mathcal{E}_N}
&\stackrel{\text{(ML)}}{\leq}
\pbb{\mathcal{E}_{Disc}}+
\sum_{\kappa=0}^{\infty}
\pbb{\mathcal{E}_{Den}\vert \bar{\mathcal{E}}_{Disc},\kappa}
\pbb{\kappa}
\nonumber \\
&\leq
\inf_{D,d}~
\frac{G}{d}\left\{
\sum_{m=1}^{\binom{M}{2}}
\binom{\binom{M}{2}}{m}\left(-1\right)^{m-1}
e^{-p\left(D-d\right)\left(1-\eta^m\right)}+
\right.
\nonumber \\
&
\hspace{1.9cm}
\left.
\sum_{i=1}^{M\kappa}
\sum_{\kappa=0}^{\infty}
\frac{\left(pD\right)^{\kappa}}{\kappa!}e^{-pD}
\binom{M\kappa}{i}
e^{-\lambda M\left(L-D\right)\left(1-e^{-\mathcal{D}_i\left(\epsilon,\kappa\right)}\right)}
\right\}
\label{eq:NoisyUpperBoundML}
\end{align}
where the minimization is constrained to $0\leq d<D\leq L$. Evidently, $\pbb{\kappa}$ is supposed to be {\it {Poisson}} distribution with parameter $pD$. In the following, we first derive the optimal $D^*$ and $d^*$ parameters which minimize the upper bound in the asymptotic regime.

In fact, when $L\gg1/p\left(1-\eta\right)$, the term corresponding to $\pbb{\mathcal{E}_{Disc}}$ can be consistently approximated by the first summand in the summation of \eqref{eq:NoisyUpperBoundML}, i.e. $m=1$. Same argument holds for the second term corresponding to denoising error bound in \eqref{eq:NoisyUpperBoundML}. As a result, the simplified upper bound on $\pbb{\mathcal{E}_N}$ in the asymptotic case can be written as:
\begin{equation}
\pbb{\mathcal{E}_N}\leq
\inf_{D,d}~
\frac{G}{d}\left\{
\binom{M}{2}e^{-p\left(1-\eta\right)\left(D-d\right)}
+
pMD
e^{-\lambda M\left(L-D\right)\left(1-e^{-\mathcal{D}_1\left(\epsilon\right)}\right)}
\right\}.
\end{equation}

An intuitive investigation of the upper bound in \eqref{eq:NoisyUpperBoundML} reveals that minimization with respect to parameters $D$ and $d$ has a non-trivial solution since all the terms $d$, $D-d$ and $L-D$ must become large in order to lower the error. Taking derivatives with respect to $d$ and $D$ and solving for the equations result in the following approximations for $D^*$ and $d^*$:
\begin{align}
D^* &\simeq
\frac{L + 
\frac{1}
{\lambda M\left(1-e^{-\mathcal{D}_1\left(\epsilon\right)}\right)}
\log\left(
\frac{\left(M-1\right)\left(1-\eta\right)}{2\left(1+ML\lambda\right)}
\right)
}
{1+\frac{p\left(1-\eta\right)}{\lambda M\left(1-e^{-\mathcal{D}_1\left(\epsilon\right)}\right)}}
\nonumber \\
d^*&\simeq
\frac{1}{p\left(1-\eta\right)}
\left(
1+
\frac{pD^* e^{-\lambda M\left(L-D^*\right)\left(1-e^{-\mathcal{D}_1\left(\epsilon\right)}\right)}}
{\left(\frac{M-1}{2}\right)e^{1-p\left(1-\eta\right)D^*}}
\right).
\label{eq:OptimalDd}
\end{align}
The relations imply that for large values of $L$ and $\lambda$, we have $D^*\rightarrow L$ and $d^{*}\rightarrow 1/p\left(1-\eta\right)$.

It should be highlighted that in the case of ML denoising, if one chooses a sufficiently large DNA read density $\lambda$, then the term associated with denoising goes to zero independent of sequencing error $\epsilon<0.5$. As a result, the upper bound of error solely depends on the discrimination condition, which for large $D$ is highly close to the upper bound of the noiseless case in Section \ref{sec:noiseless} when we choose $d^*\simeq 1/p\left(1-\eta\right)$.

\subsubsection{Spectral Denoising}

In the case of employing spectral denoising algorithm, one achieves the following upper bound on the error assembly $\pbb{\mathcal{E}_N}$:
\begin{align}
\pbb{\mathcal{E}_N}
&\stackrel{\text{(SD)}}{\leq}
\pbb{\mathcal{E}_{Disc}}+
\sum_{\kappa=0}^{\infty}
\pbb{\mathcal{E}_{Den}\vert \mathcal{E}^c_{Disc}}
\pbb{\kappa}
\nonumber \\
&\leq
\inf_{D,d}~
\frac{G}{d}\left\{
\sum_{m=1}^{\binom{M}{2}}
\binom{\binom{M}{2}}{m}\left(-1\right)^{m-1}
e^{-p\left(D-d\right)\left(1-\eta^m\right)}+
MDp e^{-\lambda MD\left(1-\exp\left(\frac{-1}{8\epsilon\left(1-\epsilon\right)}\right)\right)}+
\right.
\nonumber \\
&
\hspace{1.9cm}
\left.
\lambda M\left(L-D\right)
\sum_{\kappa=0}^{\infty}
e^{-pD}\frac{\left(pD\right)^\kappa}{\kappa!}
e^{-\zeta\left(\kappa\right)}
e^{-\lambda M\left(L-D\right)\left(1-e^{-\zeta\left(\kappa\right)}\right)}
\right\}
\label{eq:NoisyUpperBoundSD}
\end{align}
where (SD) represents the association with spectral denoising and $\zeta\left(\kappa\right)$ is defined as in \eqref{eq:ZetaDef}.

When $pD$ is sufficiently large, one can simply replace the averaging over $\kappa$ in \eqref{eq:NoisyUpperBoundSD} with $\kappa=pD$. Solving for optimal $D^*$ and $d^*$ in this case is more complicated, however, numerical results indicate close values to those obtained in the ML denoising case in \eqref{eq:OptimalDd}.


\section{Conclusions}
\label{sec:conclusion}
In this paper, the information theoretic limits of tagless pooled-DNA sequencing have been achieved and discussed. Pooled-DNA sequencing is gaining wide-spread attention as a tool for massive genetic sequencing of individuals. Moreover, the problem of Haplotype-Phasing, which is addressed in this paper as a special case of tagless pooled-DNA sequencing is an interesting area of research in many fields of bioinformatics. We have mathematically shown that the emergence of long DNA read technologies in recent years has made it possible to  phase the genomes of multiple individuals without relying on the information associated with linkage disequilibrium (LD). This achievement is critically important since assessment of LD information is expensive, and can only be exploited for the inference of DNA haploblocks, rather than whole genomes. Moreover, LD information are different from one specie to another; and even within a particular specie, from one sub-population to another one.  

We have derived necessary and sufficient conditions for whole genome assembly (phasing) when DNA reads are noiseless. Extensive theoretical analysis have been proposed to derive asymptotically tight upper and lower bounds on the assembly error rate for a given genetic and experimental setting. As a result, we have shown that genome assembly is impossible when DNA read length is lower than a specific threshold. Moreover, when DNA reads satisfy certain lengths and densities, then a simple greedy algorithm can attain the optimal result with $O\left(N\right)$ computations, where $N$ denotes the total number of DNA reads.

For the case of noisy DNA reads, a set of sufficient conditions for correct and unique assembly of all genomes are proposed whose error bounds fall close to the noiseless scenario when DNA read density is sufficiently large. We have employed two different decoding procedures to denoise DNA reads, denoted by Maximum Likelihood and Spectral denoising, respectively. Tight asymptotic upper bounds are provided for the former method, while an approximate analysis is given in the latter case which is based on recent advances in community detection literature.

In our future works we will focus on deriving both necessary and sufficient conditions in the noisy case to present tight error bounds, in addition to the optimal assembly algorithm. Moreover, the scenarios where SNP values are statisitcally linked (which resemble more realistic cases in real world applications) are among the other possible approaches for research in this area.
\section*{Appendix}
\appendix

\section{Exact Bridging Error Probability for Two Individuals}
\label{app:exactBridging}

In this section, we present the exact error probability in bridging of all identical regions for the case of two individuals. The formulation of $\pbb{\mathcal{E}_B}$ is based on mathematical analysis of bridging consecutive identical regions via Markovian random process models. In this regard, one should consider the arrivals of both DNA reads and discriminating SNPs to investigate the information flow from each discriminating SNP to its next. This procedure is described as a foregoing process which is described as follows.

We start with the first discriminating SNP in genome and denote it as the {\it {current SNP}}. Also, we denote the last DNA read in genome which includes the current SNP as the {\it {current read}}. In accordance with these definitions, one of the following arguments hold:
\begin{itemize}
\item[1)]
The current SNP does not exist.
\item[2)]
The current SNP exists, but the current read does not exist.
\item[3)]
Both the current SNP and read exist, and the current read contains the last discriminating SNP in genome as well.
\item[4)]
None of the above.
\end{itemize}
Obviously, occurrence of any of the conditions (1) or (3) indicates that the bridging condition holds, i.e. $\mathcal{E}_B$ does not occur. On the other hand, occurrence of condition (2) indicates a failure in bridging condition, since at least one identical region cannot be bridged. In the case of condition (4), let us denote the new current read as the last DNA read that starts strictly after the current read and includes the last discriminating SNP in the current read. In this regard, the new current read again falls into one of the conditions described so far. This procedure continues until bridging condition fails or is satisfied. In the following, we model the above-mentioned procedure via a pair of Markovian random processes, facilitating the mathematical formulation of $\pbb{\mathcal{E}_B}$.

Let us assume at least two discriminating SNPs exist in genome. Also, at least one DNA read exists that includes the first discriminating SNP. We denote the last read containing this SNP as the current read. Let us define two coupled and finite Markovian random processes $\ell_n$ and $d_n$, $n=1,2,\ldots,n_F$ as follows. Let $d_n$ be the distance of last discriminating SNP in the current read from the read's end position. In order to define $\ell_i$, we consider the last DNA read starting strictly after the current read which includes the last discriminating SNP in the current read. Then, $\ell_n$ represents the distance of this read from the current read's starting position. Based on these definitions, the probability of failure $P^{\left(n\right)}_{\text {Fail}},n=0,1,\ldots,n_F$ is defined as the probability that one cannot find a pair $\left(d_n,\ell_n\right)$, given that $\left(d_{n-1},\ell_{n-1}\right)$ exists. Therefore:
\begin{align}
P^{\left(n\right)}_{\text {Fail}}\vert d_{n-1},\ell_{n-1}=&~
e^{-p\left(1-\eta\right)\left(d_{n-1}+\ell_{n-1}\right)}+
\int_{0}^{d_{n-1}+\ell_{n-1}}p\left(1-\eta\right)e^{-p\left(1-\eta\right)x}
e^{-2\lambda\left(d_{n-1}+\ell_{n-1}-x\right)}{\text d}x
\nonumber \\
=&~\frac{
2\lambda e^{-p\left(1-\eta\right)\left(d_{n-1}+\ell_{n-1}\right)}-
p\left(1-\eta\right) e^{-2\lambda\left(d_{n-1}+\ell_{n-1}\right)}}{2\lambda-p\left(1-\eta\right)}.
\end{align}
The processes $d_n$ and $\ell_n$ terminate at step $n$ with probability $P^{\left(n\right)}_{\text {Fail}}$, otherwise they continue to the next step. At each step, $d_n$ and $\ell_n$ are sampled as follows:
\begin{align}
d_n~\sim~&\frac{p\left(1-\eta\right)e^{-p\left(1-\eta\right)d_n}}{1-e^{-p\left(1-\eta\right)\left(d_{n-1}+\ell_{n-1}\right)}}~,~0\leq d_n \leq d_{n-1}+\ell_{n-1}
\nonumber \\[2mm]
\ell_{n}\vert d_n~\sim~&
\frac{2\lambda e^{-2\lambda\left(L-d_{n}-\ell_n\right)}}
{1-e^{-2\lambda\left(\ell_{n-1}+d_{n-1}-d_n\right)}}
~,~
L-\ell_{n-1}-d_{n-1}\leq \ell_{n} \leq L-d_n.
\end{align}
If at the time of failure, all discriminating SNPs are covered by the sequence of current reads, then bridging condition is satisfied, otherwise it has failed. Therefore, the probability of error in bridging all identical regions can be formulated as:
\begin{equation}
\pbb{\mathcal{E}_B}=
\left(1-
e^{-Gp\left(1-\eta\right)}\left(1+Gp\left(1-\eta\right)\right)\right)
\left(1-\mathbb{E}_{L_R}\left\{\pbb{\sum_{n=1}^{n_F}\ell_n+L < L_R}\right\}\right),
\label{eq:exactBC}
\end{equation}
where $L_R$ indicates the distance between the first and the last discriminating SNPs. Although the relation in \eqref{eq:exactBC} does not give an explicit mathematical formulation for $\pbb{\mathcal{E}_B}$, it provides a rigorous numerical procedure to approximate the error probability of the bridging condition with an arbitrary high precision. Moreover, upper and lower bounds obtained in Section \ref{sec:noiseless} can be reattained by bounding the formulation in \eqref{eq:exactBC}.

\section{Proof of Lemmas in Theorem \ref{thm:NoisyTheorem}}
\label{app:Dmin}
In this section, we present the proofs of Lemmas \ref{lemma:positiveExponent}, \ref{lemma:uniqueness} and \ref{lemma:sampleDmin}.

\begin{proof}[Proof of Lemma \ref{lemma:positiveExponent}]
Maximization of \eqref{eq:Lemma1eq} implies the minimization of $J\left(s\right)\triangleq \sum_{i}P^{s+1}_iQ^{-s}_i$. For $s>0$, one can show the following inequality holds:
\begin{equation}
J\left(s\right)=\sum_{i}\frac{P_i^{s+1}}{Q_i^s}\ge
\sum_{i}P_i\left(1+\log\left(\frac{P_i}{Q_i}\right)\right)=
1+\mathcal{D}_{\text {KL}}\left(P\Vert Q\right)\ge 1.
\end{equation}
Also, again for $s<-1$ the following is true:
\begin{equation}
J\left(s\right)=\sum_{i}Q_i\frac{Q_i^{\vert s\vert-1}}{Q_i^{\vert s\vert -1}}\ge
1+\mathcal{D}_{\text {KL}}\left(Q\Vert P\right)\ge 1.
\end{equation}

As a result, the optimal $s$, denoted by $s^*$, which minimizes $J$ must be within the range $-1<s^*<0$. According to this assumption, derivative of $J$ with respect to $s$ at the optimal point results in the following equation:
\begin{equation}
\frac{\text{d} J\left(s\right)}{\text{d}s}\bigg\vert_{s=s^*}=
\sum_{i}P_i^{1-\vert s^*\vert}Q_i^{\vert s^*\vert}\log\left(\frac{P_i}{Q_i}\right)=0.
\label{eq:lemma1:derivative}
\end{equation}
It has been shown by \cite{gallager1968information} that for a broad range of finite distribution functions $P$ and $Q$, the solution of \eqref{eq:lemma1:derivative} with respect to $-1<s^*<0$ corresponds to $s^*=-1/2$. Consequently, we can use the Cauchy-Schwarz inequality to show that $\min_s~J\left(s\right)\leq1$:
\begin{equation}
\sum_{i}\sqrt{P_i Q_i}\leq
\sqrt{\left(\sum_{i}P_i\right)\left(\sum_{i}Q_i\right)}=1,
\end{equation}
where the equality holds if and only if $P_i=Q_i$ for all $i$.
\end{proof}
\begin{proof}[Proof of Lemma \ref{lemma:uniqueness}]
We will show that when $\psi, \psi'\in\Psi_{\kappa}$ and $\psi\neq\psi'$, then there exists $\phi_i\in\Phi_{\kappa}$ such that $\pbb{\phi_i\vert\psi}\neq\pbb{\phi_i\vert\psi'}$. The proof is by contradiction. Assume two sets $\psi,\psi'\in\Psi_{\kappa}$ that result in the same statistical distribution over all sequences in $\Phi_{\kappa}$. Mathematically speaking:
\begin{gather}
\pbb{\phi_i\vert \psi}=\pbb{\phi_i\vert \psi'}~,~\forall i\in\left\lbrace1,2,\ldots,2^{\kappa}\right\rbrace,
\end{gather}
or alternatively:
\begin{gather}
\sum_{j=1}^{M} x^{\rho_{ij}}-x^{\rho'_{ij}}=0,~\forall i\in\left\lbrace1,2,\ldots,2^{\kappa}\right\rbrace,
\label{eq:polyNomial}
\end{gather}
where $x\triangleq\frac{\epsilon}{1-\epsilon}$. $\rho_{ij}$ and $\rho'_{ij}$ represent the {\it {Hamming}} distances of $\phi_i$ from the $j$th SNP sequence in $\psi$ and $\psi'$, respectively. By assuming $x=1$ ($\epsilon=0.5$), all the equations in \eqref{eq:polyNomial} hold regardless of the choice of $\psi$ and $\psi'$. However, if $0\le\epsilon<0.5$ then $0\le x<1$. 

Each equation in \eqref{eq:polyNomial} is a polynomial of degree at most $\kappa$, and therefore has at most $\kappa$ roots. Acceptable roots must be real and fall in the interval $\left[0,1\right)$. Moreover, the roots must be common among all $2^{\kappa}$ equations. Consequently, one can deduce that the constraints hold only when all the coefficients of the polynomial become zero. For any $\phi_i\in\Phi_{\kappa}$, if $\phi_i\in \psi$ then one of the {\it {Hamming}} distances $\rho_{ij}$ becomes zero and one  $x^0=1$ term appears in the $i$th equation. Since the equation holds globally, i.e. for all $0\le x<1$, then a term $-x^0=-1$ must also appear in the summation. This implies that for one $j\in\left\{1,2,\ldots,M\right\}$, $\rho'_{ij}$ is zero implying $\phi_i\in\psi'$. Consequently, $\forall \phi_i\in \psi\Rightarrow\phi_i\in \psi'$ and since $\vert\psi \vert=\vert \psi' \vert=M$, we can deduce that $\psi=\psi'$. This result leads to the fact that each $\mathcal{D}\left(\epsilon\right)$ is strictly positive when $\epsilon<0.5$.
\end{proof}
\begin{proof}[Proof of Lemma \ref{lemma:sampleDmin}]
We calculate $\mathcal{D}_1\left(\epsilon\right)$ for $M=2$ and $3$ by using Lemma \ref{lemma:positiveExponent} and discussions made in Theorem \ref{thm:NoisyTheorem}. In this regard, $\mathcal{D}_1\left(\epsilon\right)$ can be formulated as
\begin{align}
\mathcal{D}_1\left(\epsilon\right)&=
\sup_{s\in\mathbb{R}}
\min_{\psi,\psi'\in\Psi_{\kappa}\atop \psi\neq\psi'}
\log
\left[
\left(
\sum_{\phi\in\Phi_{\kappa}}
\frac{
\left(\pbb{\phi\vert \psi}\right)^{s+1}
}{
\left(\pbb{\phi\vert \psi'}\right)^{s}
}
\right)^{-1}
\right]
\nonumber \\
&=
\min_{\psi,\psi'\in\Psi_{\kappa}\atop \psi\neq\psi'}
\log
\left[
\left(
\frac{\left(1-\epsilon\right)^{\kappa}}{M}
\sum_{i=1}^{2^{\kappa}}
\sqrt{\sum_{j,j'=1}^{M}
\left(\frac{\epsilon}{1-\epsilon}\right)^{\rho_{i,j}+\rho'_{i,j'}}}
\right)^{-1}
\right].
\label{eq:lemma3:Dmin}
\end{align}
In the following we compute an explicit formulation for \eqref{eq:lemma3:Dmin}.

\paragraph{Two individuals}
Evidently, the worst case analysis for $M=2$ which gives the minimal $\mathcal{D}_1\left(\epsilon\right)$ corresponds to a setting where $\psi_T$ is composed of two adjacent $\kappa$-ary SNP sequences and $\psi\in\Psi_{\kappa}^{\left(i\right)}$ differs only in one locus with $\psi_T$. From a geometric point of view, members of $\Phi_{\kappa}$ form the vertices of a $\kappa$-dimensional hyper-cube. A hyper-cube is a	 symmetric structure. Hence, without loss of generality, we can assume coordinations of $\psi_T$ and $\psi\in\Psi_{\kappa}^{\left(i\right)}$ as follows:
\begin{equation}
\psi_T=\left\{\begin{array}{c}
\left[0,0,\boldsymbol{0}_{1\times\left(\kappa-2\right)}\right]
\\[2mm]
\left[1,0,\boldsymbol{0}_{1\times\left(\kappa-2\right)}\right]
\end{array}\right\}\quad,\quad
\psi=\left\{\begin{array}{c}
\left[0,0,\boldsymbol{0}_{1\times\left(\kappa-2\right)}\right]
\\[2mm]
\left[0,1,\boldsymbol{0}_{1\times\left(\kappa-2\right)}\right]
\end{array}\right\},
\end{equation}
where $\boldsymbol{0}_{1\times\left(\kappa-2\right)}$ denotes a row vector of zeros with length $\kappa-2$.

In order to compute \eqref{eq:lemma3:Dmin}, we have to sum $\sqrt{\sum_{j,j'=1}^{2}x^{\rho_{i,j}+\rho'_{i,j'}}}$ over all vertices of hyper-cube, where $x\triangleq \frac{\epsilon}{1-\epsilon}$, and $\rho_{i,j}$ and $\rho'_{i,j'}$ denote the {\it {Hamming}} distances of the $i$th vertice from the $j$th sequence in $\psi_T$ and $j'$th sequence in $\psi$, respectively. 

Focusing on the last $\kappa-2$ coordinates of each vertice, $4\binom{\kappa-2}{h}$ vertices of the hyper-cube have a {\it {Hamming}} distance of $h$ with each of the sequences in $\psi_T$ and $\psi$  for all $h\in\left\{0,1,\ldots \kappa-2\right\}$. For the first two coordinates which separates the vertices of $\psi_T$ and $\psi$ from each other, one can simply count the differences to attain analytic values for $\rho_{i,j}$ and $\rho'_{i,j'}$. In this regard, it can be shown that:
\begin{align}
e^{-\mathcal{D}_1\left(\epsilon\right)}&=
\frac{\left(1-\epsilon\right)^{\kappa}}{2}
\sum_{h=0}^{\kappa-2}\binom{\kappa-2}{h}x^h\left(
\left(1+x\right)\sqrt{\left(1+x\right)\left(1+x\right)}+
2\sqrt{\left(1+x\right)\left(x+x^2\right)}
\right)
\nonumber \\
&=\frac{1}{2}+\sqrt{\epsilon\left(1-\epsilon\right)}.
\end{align}

\paragraph{Three individuals}
In the case of $M=3$, many of the arguments are the same to those of $M=2$ scenario. Therefore, one just needs to specify the worst case opponent subsets, i.e. $\psi_T$ and $\psi$ when there are three underlying individuals. Similar to the previous part and taking into account the symmetry property of $\kappa$-dimensional hyper-cubes, it can be shown that the worst case $\psi_T$ and $\psi$ can be chosen as follows:

\begin{equation}
\psi_T=\left\{\begin{array}{c}
\left[0,0,\boldsymbol{0}_{1\times\left(\kappa-2\right)}\right]
\\[2mm]
\left[1,0,\boldsymbol{0}_{1\times\left(\kappa-2\right)}\right]
\\[2mm]
\left[1,1,\boldsymbol{0}_{1\times\left(\kappa-2\right)}\right]
\end{array}\right\}\quad,\quad
\psi=\left\{\begin{array}{c}
\left[0,0,\boldsymbol{0}_{1\times\left(\kappa-2\right)}\right]
\\[2mm]
\left[1,0,\boldsymbol{0}_{1\times\left(\kappa-2\right)}\right]
\\[2mm]
\left[0,1,\boldsymbol{0}_{1\times\left(\kappa-2\right)}\right]
\end{array}\right\}.
\end{equation}
Based on the same arguments of the two individuals scenario, one can attain $\mathcal{D}_1\left(\epsilon\right)$ as follows:
\begin{gather}
e^{-\mathcal{D}_1\left(\epsilon\right)}=
\frac{\left(1-\epsilon\right)^{\kappa}}{3}
\sum_{h=0}^{\kappa-2}\binom{\kappa-2}{h}x^h\left(
\sqrt{\left(1+x+x^2\right)\left(1+2x\right)}+
\sqrt{\left(1+2x\right)\left(1+x+x^2\right)}+
\right.
\nonumber \\
\hspace{6.2cm}\left.
\sqrt{\left(2x+x^2\right)\left(1+x+x^2\right)}+
\sqrt{\left(1+x+x^2\right)\left(2x+x^2\right)}
\right)
\nonumber \\
\hspace{-1.3cm}
=
\frac{2}{3}\sqrt{1+\epsilon\left(1-\epsilon\right)}
\left(
\sqrt{2\epsilon\left(1-\epsilon\right)+\epsilon^2}
+
\sqrt{2\epsilon\left(1-\epsilon\right)+\left(1-\epsilon\right)^2}
\right)
\end{gather}

As can be seen, $\mathcal{D}_1\left(\epsilon\right)$ which is the dominant exponent of error for large DNA read densities, does not depend on $\kappa$, which simplifies many of the formulations in the asymptotic and total error analysis.

\end{proof}


\end{document}